\documentclass[3p]{elsarticle}

\usepackage{hyperref}

\journal{Artificial Intelligence}

\usepackage{amsmath}
\usepackage{amssymb}
\usepackage{amsthm}
\usepackage{xspace}
\usepackage{bm}
\usepackage[textwidth=16mm,textsize=footnotesize]{todonotes}
\usepackage{tikz}
\usetikzlibrary{hobby,backgrounds,positioning,calc,trees,arrows.meta}
\tikzset{> = {Latex[length=3mm,width=2mm]}}
\usepackage{boxedminipage}
\usepackage{cleveref}
\usepackage{enumitem}

\newcommand{\HDGshort}{\textsc{HDG}\xspace}
\newcommand{\MSS}{\textsc{Multidimensional Subset Sum}\xspace}
\newcommand{\MSSshort}{\textsc{MSS}\xspace}

\newcommand{\XThreeC}{\textsc{Exact Cover By $3$-Sets}\xspace}
\newcommand{\XThreeCshort}{\textsc{X3C}\xspace}
\newcommand{\PTTN}{\textsc{Partition}\xspace}

\newcommand{\allHDGNashShort}{\textsc{HDG-Nash}\xspace}

\newcommand{\allHDGIndividualShort}{\textsc{HDG-Individual}\xspace}

\newcommand{\Nat}{\mathbb{N}}

\newcommand{\bigoh}{\mathcal{O}}
\newcommand{\cc}[1]{{\mbox{\textnormal{\textsf{#1}}}}\xspace}  

\newcommand{\NP}{\cc{NP}}
\newcommand{\NPh}{\NP-hard\xspace}
\newcommand{\NPc}{\NP-complete\xspace}
\newcommand{\paraNP}{\emph{para}-\NP}
\newcommand{\FPT}{\cc{FPT}}
\newcommand{\W}[1][1]{\cc{W[#1]}}
\newcommand{\Wh}[1][1]{{\cc{W[#1]}}-hard\xspace}
\newcommand{\XP}{\cc{XP}}

\newtheorem{theorem}{Theorem}
\crefname{theorem}{Theorem}{Theorems}
\newtheorem{observation}[theorem]{Observation}
\crefname{observation}{Observation}{Observations}
\newtheorem{lemma}[theorem]{Lemma}
\crefname{lemma}{Lemma}{Lemmas}
\newtheorem{corollary}[theorem]{Corollary}
\crefname{corollary}{Corollary}{Corollaries}

\crefname{proposition}{Proposition}{Propositions}

\crefname{conjecture}{Conjecture}{Conjectures}
\newtheorem{claim}{Claim}
\crefname{claim}{Claim}{Claims}
\newenvironment{claimproof}[1]{\par\noindent\underline{Proof:}\space#1}{\hfill $\blacksquare$}

\theoremstyle{remark}
\newtheorem{example}{Example}
\crefname{example}{Example}{Examples}

\newcommand{\colors}{\ensuremath{\gamma}}
\newcommand{\coalsize}{\ensuremath{\sigma}}
\newcommand{\coalnum}{\ensuremath{{\rho_{\geq2}}}}
\newcommand{\scoalnum}{\ensuremath{{\rho_{\geq1}}}}
\newcommand{\types}{\ensuremath{\tau}}

\newcommand{\YES}{\cc{Yes}}
\newcommand{\NO}{\cc{No}}
\newcommand{\III}{\mathcal{I}}
\newcommand{\JJJ}{\mathcal{J}}
\newcommand{\PPP}{\mathcal{P}}
\newcommand{\QQQ}{\mathcal{Q}}
\newcommand{\SSS}{\mathcal{S}}

\begin{document}
	
\begin{frontmatter}
	\title{Hedonic Diversity Games: \\A Complexity Picture with More than Two Colors\tnoteref{aaai22}}
	
	\author[tuwien]{Robert Ganian}\ead{rganian@ac.wien.ac.at}
	\author[uu]{Thekla Hamm}\ead{t.l.s.hamm@gmail.com}
	\author[ctu]{Dušan Knop}\ead{dusan.knop@fit.cvut.cz}
	\author[ctu]{Šimon Schierreich}\ead{schiesim@fit.cvut.cz}
	\author[ctu]{Ondřej Suchý}\ead{ondrej.suchy@fit.cvut.cz}
	
	\address[tuwien]{Algorithms and Complexity Group, TU Wien, Austria}
	\address[uu]{Algorithms and Complexity Group, Utrecht University, The Netherlands}
	\address[ctu]{Faculty of Information Technology, Czech Technical University in Prague, Czechia}
	
	\tnotetext[aaai22]{An extended abstract of this work has been published in the proceedings of the Thirty-Sixth AAAI Conference on Artificial Intelligence (AAAI '22) \cite{GanianHKSS2022}.}
	
	\begin{abstract}
		Hedonic diversity games are a variant of the classical hedonic games designed to better model a variety of questions concerning diversity and fairness. Previous works mainly targeted the case with two diversity classes (represented as colors in the model) and provided some initial complexity-theoretic and existential results concerning Nash and individually stable outcomes. Here, we design new algorithms accompanied with lower bounds which provide a comprehensive parameterized-complexity picture for computing Nash and individually stable outcomes with respect to the most natural parameterizations of the problem. Crucially, our results hold for general hedonic diversity games where the number of colors is not necessarily restricted to two, and show that---apart from two trivial cases---a necessary condition for tractability in this setting is that the number of colors is bounded by the parameter. Moreover, for the special case of two colors we resolve an open question asked in previous work~(Boehmer and Elkind, AAAI 2020).
	\end{abstract}
	
	
\end{frontmatter}

\section{Introduction}

Settings in which individual agents form groups or \emph{coalitions} are ubiquitous in scenarios which are modeled by computational social choice theory and arise throughout social, political and economic life.
Modeling and investigating the behavior of agents in these kinds of situations is the main goal of the field of cooperative game theory~(see \cite{ChalkiadakisEW11} for a survey with an algorithmic focus), and an important subarea concerns the study of which coalitions are formed by agents based on their various preferences.
One of the most simple yet compelling models for the formation of coalitions is based on the assumption of \emph{hedonic} behavior which postulates that an agent's preference is based purely on the composition of their own coalition, and ignores any potential inter-coalitional or global effects~\cite{dreze1980hedonic,ElkindFF20,KerkmannLRRSS20,0001BW21,BanerjeeKS01}. Anonymous hedonic games~\cite{BogomolnaiaJ2002} are a well-studied restriction of hedonic games where preferences only take into account the sizes of the coalitions, not the precise identity of their members.

Bredereck et al.~\cite{BredereckEI2019} recently initiated the study of a model that naturally falls between general hedonic games and anonymous hedonic games in terms of to what level of specificity agents are distinguished by each other's preferences. Specifically, they considered the scenario where each agent is assigned a class or \emph{color}, and preference profiles are defined with respect to the ratios of classes occurring in the coalition. Loosely following their terminology, we call these \emph{hedonic diversity games}\footnote{This term was originally used for the special case of two classes, and the presence of $k>2$ classes was previously identified by adding ``\emph{$k$-tuple}'' or ``\emph{k-}''\cite{BoehmerE2020}. Since here we consider $k$ to be part of the input, we use \emph{hedonic diversity games} for the general model.}.
The question targeted by hedonic diversity games occurs in a number of distinct settings (which are not captured by anonymous hedonic games), such as the Bakers and Millers game (where coalitions are formed between two types of agents with competing preferences)~\cite{AzizBBHOP19,schelling1971dynamic,Bilo0FMM18} or when the task is to assign guests from several backgrounds to tables~\cite{BoehmerE2020:ijcai,IgarashiSZ19}.

The most prominent computational question arising from the study of hedonic diversity games targets the computation of an outcome that is stable under some well-defined notion of stability. While several stability concepts have been considered in the literature including, e.g., \emph{envy-freeness} and \emph{core stability}, here we focus on two prominent notions of single-agent stability that have been applied in the context of hedonic games:
\begin{itemize}
	\item \emph{Nash stability}, which ensures that no agent prefers leaving their coalition to join a different one, and
	\item \emph{individual stability}, where no agent prefers leaving their coalition to join another coalition whose members would all appreciate or be indifferent to such a change.
\end{itemize}

In their pioneering work, Boehmer and Elkind~\cite{BoehmerE2020} have provided initial results on the computational complexity of computing stable outcomes for hedonic diversity games (a problem we hereinafter generally refer to as \HDGshort), with a particular focus on the case with $2$ colors. They analyzed the problem not only from the viewpoint of classical complexity, but also with respect to the more refined \emph{parameterized} paradigm\footnote{A brief introduction to parameterized complexity is provided in the Preliminaries.}. Among others, they showed that individually and Nash-stable \HDGshort\ is \NP-hard even when restricted to instances with $5$ and $2$ colors, respectively. Moreover, for the case of $2$ colors they obtained a polynomial time algorithm for individually stable \HDGshort\ and an \XP algorithm for Nash-stable \HDGshort\ when parameterized by the size of the smaller color class~\cite[Theorem~5.2]{BoehmerE2020}. Still, our understanding of the computational aspects of \HDGshort has up to now remained highly incomplete: no tractable fragments of the problem have been identified beyond the aforementioned \XP\ algorithm targeting Nash stability for \(2\) colors and the polynomial-time algorithm for individual stability for \(2\) colors. In fact, even whether this known \XP-algorithm for \(2\) colors can be improved to a fixed-parameter one has been explicitly stated as an open problem~\cite[Section 3]{BoehmerE2020}.

\paragraph{Our Contribution} 
We obtain new algorithms and lower bounds that paint a comprehensive picture of the complexity of \HDGshort through the lens of parameterized complexity. In particular, our results provide a complete understanding of the exact boundaries between tractable and intractable cases for \HDGshort for both notions of stability and with respect to the most fundamental parameters that can either be identified from the input or have been used as additional conditions defining which outcomes are accepted:

\begin{itemize}
  \item The number of colors on the input, denoted by $\colors$;
  \item The maximum size $\coalsize$ of a coalition in an acceptable outcome. This is motivated by application scenarios as well as several works on the stable roommates problem---a special case of hedonic games---restricted to coalitions of fixed size~\cite{Huang07,NgH91,BredereckHKN20};
  \item Two possible bounds on the maximum number of coalitions in an acceptable outcome: $\scoalnum$ bounds the total number of coalitions \emph{including} agents who are alone, while $\coalnum$ bounds the total number of coalitions \emph{excluding} agents who are alone. Both of these restrictions induce different complexity-theoretic behaviors and are meaningful in different contexts---for instance, the former when coalitions represent tables at a conference dinner, while the latter when coalitions represent optional sports activities.
  \item The number $\types$ of \emph{agent types} of the instance, i.e., the number of different preference lists that need to be considered. Agent types have been used as a parameter in many related works~\cite{BredereckHKN20,BiroIS11,KavithaNN14,IrvingMS08} and is a more relaxed parameterization than the number of agents. In the context of \HDGshort, the number of agents bounds the size of the instance and hence is not an interesting parameter.
\end{itemize}

The complexity picture for \HDGshort\ is provided in \Cref{fig:complexityPicture} and is based on a set of three non-trivial algorithms and four reductions. Our results show that apart from two trivial cases (restricting the size and number of coalitions), a necessary condition for tractability under the considered parameterizations is that the number of colors is bounded by the parameter---and this condition is also sufficient, at least as far as \XP-tractability is concerned. Remarkably, in our general setting the problem retains the same complexity regardless of whether we aim for Nash or individual stability---this contrasts the previously studied subcase of $\gamma=2$~\cite{BoehmerE2020}.

Apart from requiring non-trivial insight into the structure of a possible solution, the approaches used to solve the tractable fragments vary greatly from one another: the fixed-parameter algorithm w.r.t.\ $\colors+\coalsize$ combines branching with an ILP formulation, the \XP-algorithm w.r.t.\ $\colors+\types$ relies on advanced dynamic programming and the \XP-algorithm which parameterizes by $\colors+\coalnum$ combines careful branching with a network flow subroutine.
We complement these positive results with lower bounds which show that none of the parameters can be dropped in any of the algorithms. This is achieved by a set of \W[1]-hardness and \NP-hardness reductions, each utilizing different ideas and starting from diverse problems: \textsc{Independent Set}, \textsc{Multidimensional Subset Sum}, \textsc{Exact Cover by 3-Sets} and \textsc{Partition}.

As our final result, we resolve an open question of Boehmer and Elkind~\cite{BoehmerE2020}: when the number of colors ($\gamma$) is $2$, is Nash-stable \HDGshort\ fixed-parameter tractable when parameterized by the size of the smaller color class? Here, we provide a highly non-elementary reduction from a variant of the \textsc{Group Activity Selection} problem~\cite{DarmannDDLS17,EibenGO18} which excludes fixed-parameter tractability.

\begin{figure}[bt]
	\centering
	\begin{tikzpicture}[node distance=1cm]
		\tikzstyle{every node} = [minimum height=0.6cm,,minimum width=0.7cm]
		\tikzstyle{result} = [draw,thick]
		\tikzstyle{known} = [label={[label distance=-5pt]270:\small{\tt Known}}]
		\tikzstyle{NPh} = [fill=red!30]
		\tikzstyle{Wh} = [fill=orange!30]
		\tikzstyle{FPT} = [fill=green!30]

		\node[NPh,known] (C) at (0, 0) {$\colors$};
		\node[NPh] (SN) [right=of C] {$\scoalnum$};
		\node[NPh] (S) [right=of SN] {$\coalsize$};
		\node[NPh] (T) [right=of S] {$\types$};
		\node[NPh] (WN) [right=of T] {$\coalnum$};

		\node[result,FPT,label={[label distance=-5pt]90:\small{\tt T\ref{lem:allHDG:FPT:ColorsSize}}}] (CS) [below left=1 and 0 of C] {$\colors\coalsize$};
		\node[Wh,result,label={[label distance=-5pt]90:\small{\tt C\ref{lem:allHDGNash:XP:ColorsStrongNum}}}] (CSN) [right=0.15 of CS] {$\colors\scoalnum$};
		\node[Wh,result,label={[label distance=-5pt]90:\small{\tt T\ref{lem:allHDG:XP:ColorsTypes}}}] (CT) [right=0.15 of CSN] {$\colors\types$};
		\node[Wh,result,label={[label distance=-5pt]90:\small{\tt T\ref{lem:allHDGNash:XP:ColorsWeakNum}}}] (CWN) [right=0.15 of CT] {$\colors\coalnum$};
		\node[result,NPh,label={[label distance=-5pt]270:\small{\tt\, T\ref{lem:allHDG:NPh:TypesNum}}}] (SNT) [right=0.15 of CWN] {$\scoalnum\types$};
		\node[result,NPh,label={[label distance=-5pt]270:\small{\tt T\ref{lem:allHDG:NPh:SizeTypes}}}] (ST) [right=0.15 of SNT] {$\coalsize\types$};
		\node[result,NPh,label={[label distance=-5pt]270:\small{\tt T\ref{lem:allHDG:NPh:TypesNum}}}] (WNT) [right=0.15 of ST] {$\coalnum\types$};
		\node[result,FPT,label={[label distance=-5pt]90:\small{\tt L\ref{lem:allHDG:FPT:StrongNumSize}}}] (SSN) [right=0.15 of WNT] {$\coalsize\scoalnum$};
		\node[Wh,result,label={[label distance=-5pt]90:\small{\tt L\ref{lem:allHDG:XP:WeakNumSize}}}] (SWN) [right=0.15 of SSN] {$\coalsize\coalnum$};

		\node[FPT] (CSSN) [below=1 of CS] {$\colors\coalsize\scoalnum$};
		\node[FPT] (CSWN) [right=0.25 of CSSN] {$\colors\coalsize\coalnum$};
		\node[FPT] (CST) [right=0.25 of CSWN] {$\colors\coalsize\types$};
		\node[result,Wh,label={[label distance=-5pt]270:\small{\tt T\ref{lem:allHDG:Wh:ColorsNumTypes}}}] (CSNT) [right=0.25 of CST] {$\colors\scoalnum\types$};
		\node[Wh] (CWNT) [right=0.25 of CSNT] {$\colors\coalnum\types$};
		\node[FPT] (SSNT) [right=0.25 of CWNT] {$\coalsize\scoalnum\types$};
		\node[result,Wh,label={[label distance=-5pt]270:\small{\tt T\ref{lem:allHDG:Wh:SizeTypeNum}}}] (SWNT) [right=0.25 of SSNT] {$\coalsize\coalnum\types$};

		\node[FPT] (CSSNT) [below right=1 and 0.25 of CSWN] {$\colors\coalsize\scoalnum\types$};
		\node[FPT] (CSWNT) [right=of CSSNT] {$\colors\coalsize\coalnum\types$};

		\draw[->] (SSN.south) -- (SSNT.north);
		\draw[->] (CS.south) -- (CSSN.north);
		\draw[->] (CS.south) -- (CSWN.north);
		\draw[->] (CS.south) -- (CST.north);
		\draw[->] (CST.south) -- (CSSNT.north);
		\draw[->] (CST.south) -- (CSWNT.north);

		\draw[->] (ST.north) -- (S.south);
		\draw[->] (ST.north) -- (T.south);
		\draw[->] (SNT.north) -- (SN.south);
		\draw[->] (WNT.north) -- (WN.south);

		\draw[<->] (CSNT.north) -- (CT.south);
		\draw[->] (CSNT.north) -- (CSN.south);
		\draw[->] (3.7,0 |- CSNT.south east) .. controls +(down:5mm) and +(down:5mm) .. (CWNT.south);
		\draw[<->] (CWNT.north) -- (CWN.south);

		\draw[<->] (SWNT.north) -- (SWN.south);
	\end{tikzpicture}
	\caption{The complexity picture for \HDGshort for both Nash and individual stability. Combinations of parameters which give rise to fixed-parameter algorithms are highlighted in green, while combinations for which \HDGshort is \Wh but in \XP are highlighted in orange and \NPc combinations are highlighted in red. Results explicitly proved in this work are represented by a black box and a reference to the given theorem, corollary or lemma.}
	\label{fig:complexityPicture}
\end{figure}

\paragraph{Related Work}
Hedonic games generalize various well-studied matching models
such as the marriage matching model~\cite{GaleS62}, the roommate matching model~\cite{GaleS62}, and many-to-one matching model~\cite{RothS90}. These are all models in which agents have hedonic preferences over the partition formed by the different coalition but where, in contrast to general hedonic games, not all partitions into coalitions are allowed.

In the extensive research on hedonic games a range of solution concepts have been considered for the outcomes of hedonic games.
Classically one requires stability of an outcome under deviation of single agents, this includes Nash and individual stability which are the solution concepts considered in this work, or groups of agents~\cite{BogomolnaiaJ02}.
Other solution concepts include individual rationality, perfection and Pareto optimality~\cite{AzizBH13}, or scores from voting theory under which an outcome should maximize social welfare and which have been considered for restricted variants of hedonic games~\cite{AzizGGMT15,AzizBS11}.
So far, research on hedonic diversity games, like our work has focused on stability with respect to single-agent deviations~\cite{BoehmerE2020,0001BW21,Darmann21} but maximizing social welfare has also been studied by Darmann~\cite{Darmann21}.

Inspired by the definition of hedonic diversity games, Boehmer and Elkind~\cite{BoehmerE2020:ijcai}, studied the \emph{roommate diversity problem} in which agents have color-fraction based preferences over their coalitions.
Even before the introduction of hedonic diversity games by Bredereck et al.~\cite{BredereckEI2019}, diversity has been studied in matching models~\cite{Huang10,KamadaK15}, with the crucial difference that in hedonic diversity games different colors are used to define agents preferences while in matching models they are used to prescribe distributional constraints on the outcome.
Considering diversity as defining feature of agents' preferences over coalitions is also related and was originally strongly motivated by its importance in residential community formation which is formalized by the theoretical model of Schelling games~\cite{Schelling71,AgarwalEGISV2021} and, more recently, in the refugee housing problem~\cite{KnopS2023}.
Important distinctions between Schelling games and refugee housing on one hand, and hedonic games on the other hand are that in the former communities, unlike coalitions in the latter, are not necessarily disjoint, and are also restricted by some underlying topology.

\section{Preliminaries}

For integers $i<j$, we let $[i]=\{1,\dotsc,i\}$, $[i]_0=[i]\cup\{0\}$ and \([i,j] = \{i, \dotsc, j\}\). $\Nat$ denotes the set of positive integers.

\subsection{Hedonic Diversity Games}
Let $N=[n]$ be a set of agents partitioned into color classes $D_1,\dots,D_\colors$. Let a \emph{palette} be a tuple of the form $(r_1,\dots,r_\colors)$ where (1) $\sum_{i\in [\colors]}r_i=1$ and (2) each $r_i$, $i\in [\colors]$, is non-negative rational number with a denominator of at most $n$. 
A subset $C\subseteq N$ is called a \emph{coalition}, and if $|C|=1$ we call it a \emph{trivial coalition}. For a coalition $C$, we define the \emph{palette of $C$} as the tuple $\left(\frac{|D_c\cap C|}{|C|}\right)_{c\in [\colors]}$. Intuitively, in the setting of hedonic diversity games agents will judge coalitions based only on their palettes (i.e., the tuples of fractions capturing how well each diversity class is represented). Formally, let $M$ be a set of weak orders over the set of all palettes, and let $\succeq$ be a (not necessarily surjective) mapping which assigns each agent $i\in N$ to its preference list $\succeq_i\in M$. We refer to the preference list $\succeq_i$ as the \emph{master list}. We remark that some palettes can never occur in a coalition containing an agent with a certain color, and such ``irrelevant'' palettes can be omitted from the preferences of that agent for brevity. Interested readers may find an example of these notions later on in this section (cf. \Cref{ref:example}).

Let $\Pi = \{C_1,\ldots,C_\ell\}$ be a partitioning of the agents, i.e., a set of subsets of $N$ such that $\bigcup_{i=1}^\ell C_i = N$ and all $C_i$'s are pairwise disjoint. We call $\Pi$ an \emph{outcome}, and use $\Pi_i$ to denote the coalition the agent $i$ is involved in for the outcome $\Pi$.
The notion of \emph{stability} of an outcome $\Pi$ for a game $(N,(\succeq_i)_{i\in N})$ will be crucial for our considerations. If there is an agent $i$ and coalition $C$ in $\Pi$ (possibly allowing $C=\emptyset$) such that $C\cup\{i\}\succ_i \Pi_i$, we say $i$ admits an \emph{NS-deviation} to~$C$. $\Pi$ is called \emph{Nash stable} (NS) if it contains no agent with an NS-deviation. If agent $i$ admits an NS-deviation to $C$ where in addition for each agent $j\in C$ it holds that $C\cup\{i\} \succeq_j C$, we say that $i$ admits an \emph{IS-deviation} to~$C$. $\Pi$ is then called \emph{individually stable} (IS) if it contains no agent with an IS-deviation. The core computational task that arises in the study of hedonic diversity games is determining whether an instance admits a stable outcome w.r.t.\ the chosen notion of stability.

The classical notion of hedonic games coincides with the special case of hedonic diversity games where each agent receives their own color.
Since already for hedonic games the preference profiles may be exponentially larger than $n$ even after removing all coalitions that are strictly less favorable than being alone, and this blow-up does not reflect the usual application scenarios for hedonic games, we adopt the \emph{oracle model} that has been proposed in previous works~\cite{Peters2016dichotomous,Peters2016graphical,HanakaL2022,IgarashiE16}: instead of having the preference profiles included on the input, we are provided with an oracle that can be queried to determine (in constant time) whether an agent $i\in N$ prefers some coalition to another coalition.

Our aim will be to obtain an understanding of the problem's complexity with respect to the following parameters and their combinations:
\begin{description}[leftmargin=!,labelindent=1.5cm,align=right]
  \item[$\colors$] the number of color classes,
  \item[$\coalsize$] the maximum size of a coalition,
  \item[$\scoalnum$] the maximum number of coalitions,
  \item[$\coalnum$] the maximum number of non-trivial coalitions,
  \item[$\types$] the number of \emph{agent types}, formally defined as $|M|$.
\end{description}

We can now formalize our problems of interest:

\begin{center}
\begin{boxedminipage}{0.98 \columnwidth}
\allHDGNashShort\\[5pt]
\begin{tabular}{l p{0.78 \columnwidth}}
Input: & Instance $\III$ consisting of a set $N=[n]$ of agents partitioned into $D_1,\dots,D_\colors$, an oracle that can compare coalitions according to $\succeq_i$ for each $i\in N$, and integers $\coalsize$, $\scoalnum$, $\coalnum$.\\
Question: \hspace{-0.4cm} & Does $\III$ admit a Nash stable outcome consisting of at most $\scoalnum$ coalitions and $\coalnum$ non-trivial coalitions, each of size at most $\coalsize$?
\end{tabular}
\end{boxedminipage}
\end{center}

\allHDGIndividualShort is then defined analogously, with the distinction that the outcome must be individually stable; we use \HDGshort\ to jointly refer to both problems. Observe that the restrictions imposed by $\coalsize$, $\scoalnum$, $\coalnum$ can be made irrelevant by setting these integers to $n$, meaning that \allHDGNashShort\ and \allHDGIndividualShort generalize the case where these three restrictions are removed.
As the complexity of our algorithms which do not use any of $\coalsize$, $\scoalnum$ or $\coalnum$ as parameter depends only polynomially on $\coalsize$, $\scoalnum$ and $\coalnum$, this means that, removing the restriction completely (by setting the respective measures to \(n\)) does not affect the algorithms' parameterized complexity.

Now, we give a simple example illustrating both stability concepts and the notion of types and colors used thorough this work. The example also illustrates the concept of \emph{master lists}. 
\begin{example}
\label{ref:example}
	Let $\III$ be an instance with $4$ agents, denoted $a,b,c,d$ for clarity, such that $D_1 = \{a,b\}$ and $D_2 = \{c,d\}$. The agents have the following preference lists:
	\begin{align*}
		a &:: \textstyle (\frac13,\frac23) \succ (\frac23,\frac13) \succ (1,0) \succ (\frac12,\frac12) \\
		b &:: \textstyle (\frac12,\frac12) \succ (\frac13,\frac23) \succ (1,0) \succ (\frac23,\frac13) \\
		c,d &:: \textstyle (\frac13,\frac23) \succ (\frac23,\frac13) \succ (\frac12,\frac12) \succ (0,1)
	\end{align*}

	One can see that the preferences of agents $a$, $c$, and $d$ are derivable from the same master list \[ \textstyle (\frac13,\frac23) \succ (\frac23,\frac13) \succ (1,0) \succ (\frac12,\frac12) \succ (0,1), \] whereas preference list of agent $b$ is not, i.e., $\III$ has exactly two types of agents and two colors. However, the color classes and type classes do not coincide.

	If we are interested in an individually stable outcome, we can choose $C_1 = \{a,c,d\}$ and $C_2 = \{b\}$. Note that this outcome is not Nash stable, since $b$'s current palette is $(1,0)$ and she prefers to be in the grand coalition $\{a,b,c,d\}$ with palette $(\frac12,\frac12)$. Moreover, in contrast to IS-stability, it does not matter whether other agents are willing to accept $b$. On the other hand, the outcome $C_1 = \{b,c,d\}$, $C_2 = \{a\}$ is Nash stable (and hence also individually stable).
\end{example}

\subsection{Parameterized Complexity}
Parameterized complexity~\cite{CyganFKLMPPS15,DowneyFellows13,Niedermeier06} analyzes the
running time of algorithms with respect to a parameter
$k\in\Nat$ and input size~$n$. The high-level idea is to find a parameter
that describes the structure of the instance such that the
combinatorial explosion can be confined to this parameter. In this
respect, the most favorable complexity class is \FPT
(\textit{fixed-parameter tractable}) which contains all problems that
can be decided by an algorithm running in $f(k)\cdot
n^{\bigoh(1)}$ time, where $f$ is a computable function. Algorithms with
this running-time are called \emph{fixed-parameter algorithms}. A less
favorable outcome is an \XP{} \emph{algorithm}, which is an algorithm
running in $n^{f(k)}$ time; problems admitting such
algorithms belong to the class \XP. Naturally, it may also happen that an \NP-complete problem remains \NP-hard even for a fixed value of $k$, in which case we call the problem \emph{\paraNP-complete}.

Showing that a problem is $\W[1]$-hard rules out the existence of a fixed-parameter algorithm under the well-established assumption that $\W[1]\neq \FPT$.
This is done via a \emph{parameterized reduction}~\cite{CyganFKLMPPS15,DowneyFellows13} from some known $\W[1]$-hard problem. A parameterized reduction from a parameterized problem $\PPP$ to a parameterized problem $\QQQ$ is a function:
\begin{itemize}
	\item which maps \YES-instances to \YES-instances and \NO-instances to \NO-instances,
	\item which can be computed in $f(k)\cdot
	n^{\bigoh(1)}$ time, where $f$ is a computable function, and
	\item where the parameter of the output instance can be upper-bounded by some function of the parameter of the input instance.
\end{itemize}

\section{Algorithms and Tractable Fragments}

The aim of this section is to establish the algorithmic upper bounds that form the foundation for the complexity picture provided in \Cref{fig:complexityPicture}. We begin with a fairly simple observation that establishes our first two tractable fragments.

\begin{lemma}\label{lem:allHDG:FPT:StrongNumSize}\label{lem:allHDG:XP:WeakNumSize}
 \allHDGNashShort and \allHDGIndividualShort are:
	\begin{enumerate}
		\item fixed-parameter tractable parameterized by $\scoalnum+\coalsize$, and
		\item in \XP parameterized by $\coalnum+\coalsize$.
	\end{enumerate}
\end{lemma}
\begin{proof}
	For the first claim, note that either $n \le \scoalnum \cdot \coalsize$ or we are facing an obvious \NO-instance. Hence we can restrict our attention to the case where $n$ is bounded by a function of  the parameters.
	Therefore, in \FPT{} time it is possible to enumerate all possible coalition structures and check their stability and properties. There are at most $\scoalnum^n = \scoalnum^{\scoalnum \cdot \coalsize}$ possible outcomes and we can check whether any such outcome is stable in time at most $\bigoh(\scoalnum^2\cdot \coalsize)$, which results in a total running time of at most $\bigoh(\coalsize\scoalnum^{\coalsize\scoalnum+2})$.

	For the second claim, we begin by branching to determine the structure of non-trivial coalitions in a solution, that is, we branch on how many non-trivial coalitions there will be in the outcome and for each such coalition we branch to determine its size. This yields a branching factor of at most $\coalnum \cdot \coalsize^\coalnum$. Now in each branch we have at most $\coalnum \cdot \coalsize$ ``positions'' for agents in the identified non-trivial coalitions, and we apply an additional round of branching to determine which agents will occupy these positions. This yields an additional branching factor of at most $n^{\coalnum \cdot \coalsize}$. At this point, the outcome is fully determined and we can check whether it is stable and satisfies the requirements in $\bigoh(n^2)$ time. This results in an algorithm with a total running time of at most $\bigoh(\coalnum\cdot\coalsize^{\coalnum}n^{\coalnum \cdot \coalsize+2})$.
\end{proof}

Our first non-trivial result is an algorithm for solving \HDGshort\ parameterized by the number of colors and the coalition size.
There, we first observe that in this case we may assume the number of types to also be bounded.

\begin{observation}\label{obs:colors_sizes_types}
	It holds that $\types \le (\colors^{\coalsize+1})!\cdot 2^{\colors^{\coalsize+1}}$.
\end{observation}
\begin{proof}
	As the size of coalitions is bounded by $\coalsize$ and there are $\colors$ colors of agents, there are $p=\sum_{x = 1}^{\coalsize} \colors^x \le \colors^{\coalsize+1}$ possible palettes.
	Therefore, there are at most $p!\cdot 2^p \le (\colors^{\coalsize+1})!\cdot 2^{\colors^{\coalsize+1}}$ weak orders over the palettes in total, yielding the bound.
\end{proof}

\begin{theorem}\label{lem:allHDG:FPT:ColorsSize}
	\allHDGNashShort and \allHDGIndividualShort are fixed-parameter tractable parameterized by $\colors+\coalsize$.
\end{theorem}
\begin{proof}
	Let $\mathcal{C}$ be the set of all possible types of coalitions, where each type of coalition $C$ is specified by the set of the combination of agent type and color of each agent in a coalition of type $C$. Observe that $|\mathcal{C}| \le (\types \cdot \colors)^{\coalsize+1}$, which by \Cref{obs:colors_sizes_types} means that $|\mathcal{C}| \le 2^{\colors^{\bigoh(\coalsize)}}$.
	Let $\mathcal{C}_{\ge 2}$ be the set of all non-trivial types of coalitions from $\mathcal{C}$.
	For each type of coalition $C \in \mathcal{C}$ we now branch to determine whether it occurs $0$, $1$, or at least $2$ times in the solution, i.e., in a hypothetical Nash-stable or individually stable outcome. For each type of coalition $C$, we store the condition on the number of its occurrences imposed by this branching via the variable $\pi_C \in \{0,1,2^+\}$; an outcome \emph{respects} the branch if each type of coalition occurs in an outcome \(\Pi\) $\pi_C$ many times. Observe that the information given by the variables $\pi_C$ is sufficient to determine whether an outcome respecting the branch will be (Nash or individually) stable. In fact:

	\begin{claim}
		It is possible to verify, in time at most $\coalsize \cdot |\mathcal{C}|^2$, whether every outcome respecting the branch $(\pi_C)_{C \in \mathcal{C}}$ is (Nash or individually) stable, or whether no such outcome is (Nash or individually) stable.
	\end{claim}
	\begin{claimproof}
		A simple exhaustive algorithm suffices.
		We loop over all types of coalitions $C \in \mathcal{C}$ with $\pi_C \ge 1$ and all its (at most~$\coalsize$) agent type and color combinations $i \in C$.
		Then we check for each coalition type $C' \in \mathcal{C} \setminus \{C\}$ with $\pi_{C'} \ge 1$ whether $(i,C,C')$ yields a deviation (either NS- or IS-deviation).
		Furthermore, if $\pi_C = 2^+$ we also check whether $(i,C,C)$ yields a deviation.
		The outcome is stable if and only if we do not find any deviation.
	\end{claimproof}

  	At this point it remains to determine whether there exists an outcome $\Pi$ respecting the branch $(\pi_C)_{C \in \mathcal{C}}$.
 	 We resolve this by encoding the question into an Integer Linear Program (ILP) with boundedly-many constraints.

	Let $a^C_{c,t}$ be the number of agents of color $c$ and type~$t$ in a type of coalition $C \in \mathcal{C}$.
  	We denote by $n_{c,t}$ the number of agents of color $c$ and type $t$ in the instance (note that the full preferences of each agent can be obtained using $\bigoh(|\mathcal{C}|^2 \cdot n)$ oracle calls).
  	Further, we denote by $\mathcal{C}'$ the set of types of coalitions from $\mathcal{C}$ with $\pi_C = 2^+$ and $\mathcal{C}'_{\ge 2}=\mathcal{C}' \cap \mathcal{C}_{\ge 2}$.
  	We construct the following ILP with integer variables $x_C$ for $C \in \mathcal{C}'$ which encode the number of times a coalition with type \(C\) occurs in our solution:

  	\begingroup
  	\allowdisplaybreaks
  	\begin{align*}
    	\sum_{C \in \mathcal{C}'} a^C_{c,t} \cdot x_C &= n_{c,t} - \sum_{C \in \mathcal{C}\setminus\mathcal{C}'}a^C_{c,t} \cdot \pi_C
                                                                     &\forall c \in [\colors],t\in[\types]  \\
                    \sum_{C \in \mathcal{C}'} x_C &\le \scoalnum - \sum_{C \in \mathcal{C} \setminus \mathcal{C}'}\pi_C       \\
            \sum_{C \in \mathcal{C}'_{\ge 2}} x_C &\le \coalnum - \sum_{C \in \mathcal{C}_{\ge 2} \setminus \mathcal{C}'}\pi_C       \\
                                              x_C &\ge 2             &\forall C \in \mathcal{C}'
  	\end{align*}
  	\endgroup
  	It is straightforward to verify that the ILP is feasible if and only if it is possible to realize an outcome respecting the branch $(\pi_C)_{C \in \mathcal{C}}$ using the given agent set:
  	Indeed, the first set of inequalities ensures the sufficient amount of agents of the correct type to realize the number of coalitions of type \(C\) described by the value of \(x_C\) in a solution of the ILP.
  	The second and third set of inequalities ensure that we do not create to many coalitions and non-trivial coalitions respectively.
  	The final set of inequalities ensures that indeed at least two coalitions of each type in \(\mathcal{C}'\) are created, conforming to our branch.
  	
  	As the largest coefficient is bounded by $\coalsize$, the right-hand sides are bounded by $n$, the number of constraints (excluding the non-negativity constraints) is bounded by $\colors\cdot\types+2$, and the number of variables is $|\mathcal{C}|=2^{\colors^{\bigoh(\coalsize)}}$, it is possible to determine whether this ILP is feasible in
  	\[\bigoh\Big(\big(\coalsize^2(\colors\cdot\types+2)\big)^{\colors\cdot\types+2}\cdot \log n + 2^{\colors^{\bigoh(\coalsize)}}\cdot(\colors\cdot\types+2)\Big)=2^{2^{\colors^{\bigoh(\coalsize)}}}\cdot \log n\] time~\cite{EisenbrandW2020, JansenR19}.
  	The total running time is hence at most $2^{2^{\colors^{\bigoh(\coalsize)}}}\cdot n$.
\end{proof}

Next, we turn our attention to \HDGshort\ parameterized by the number of colors and agent types. While our previous result reduced the problem to a tractable fragment of ILP, here we will employ a non-trivial dynamic programming subroutine.

\begin{theorem}\label{lem:allHDG:XP:ColorsTypes}
	\allHDGNashShort and \allHDGIndividualShort are in \XP parameterized by $\colors+\types$.
\end{theorem}
\begin{proof}
	For each color $c\in [\colors]$ and agent type $t$ (represented as an integer in $[\types]$), let $n_{c,t}$ be the number of agents of color $c$ and type $t$ in the instance (the type of each agent can be determined using  $\bigoh(n^{2 \colors})$ oracle calls). It will be useful to observe that agents of the same type and color are pairwise interchangeable
	without affecting the stability of an outcome.
	
	For each combination of $c\in [\colors]$ and $t\in[\types]$, we now branch to determine the ``worst'' and ``second-worst'' coalitions in which an agent of color $c$ and type $t$ appears in a sought-after NS/IS outcome. Formally, this branching identifies $(\colors\cdot \types)$-many coalitions $C_1^{c,t}$ and $C_2^{c,t}$, each represented as a palette $\left(\frac{|C_\star^{c,t}\cap D_1|}{|C_\star^{c,t}|},\dots, \frac{|C_\star^{c,t}\cap D_\colors|}{|C_\star^{c,t}|}\right)$, such that $C_2^{c,t}\succeq_t C_1^{c,t}$. Since there are at most $(n+1)^\colors$ many such palettes, the number of branches is upper-bounded by $\big((n+1)^{2\colors^2 \cdot \types}\big)$.
	
	Our aim is to determine whether it is possible to pack agents into stable coalitions, whereas the insight we use is that the information we branched on allows us to pre-determine whether a coalition will be stable or not.
	Intuitively, the coalitions should not be ``worse'' than the branch allows (only one coalition as ``bad'' as $C_1^{c,t}$ is allowed, unless $C_1^{c,t}\sim_t C_2^{c,t}$, and the others must be at least as ``good'' as $C_2^{c,t}$) ensuring that they do not become the source of a deviation and should not represent a destination for a deviation from coalitions as ``bad'' as $C_1^{c,t}$ or $C_2^{c,t}$.
	We solve this packing problem via dynamic programming.
	
	Before we proceed to formalizing our solution, we provide a high-level description of the dynamic programming procedure. The procedure constructs coalitions one by one, and stores information about the constructed coalitions in the form of ``patterns'' which capture all the required information to prevent undesirable deviations. Crucially, the number of patterns is bounded by $n^{\bigoh(\colors\cdot \types)}$, and the program stores a table which keeps track of whether each pattern can be realized or not. As soon as the algorithm constructs a pattern which corresponds to a solution we output \texttt{yes}, while if the algorithm reaches a stage where it cannot identify any new realized patterns we output \texttt{no}. 
	
	Formally, we say that a \emph{partial solution} is a packing $C_1, \ldots, C_\ell$ of coalitions such that for every $i \in [\ell]$ we have
	\begin{itemize}
		\item $C_i \subseteq N$,
		\item $|C_i| \le \coalsize$,
		\item for every color $c$ and type $t$, if $C_i$ contains an agent of color $c$ and type $t$, then $C_i\succeq_t C_1^{c,t}$,
		\item for every color $c$ and type $t$
			\begin{itemize}
				\item either $C_i$ contains all $n_{c,t}$ agents of color $c$ and type $t$,
				\item or $C_i$ contains an agent of color $c$ and type $t$, $C_i \prec_t C_2^{c,t}$, and $C_i^{+(c,t)} \preceq_t C_2^{c,t}$,
				\item or $C_i^{+(c,t)} \preceq_t C_1^{c,t}$,
				\item or alternatively, in case of IS, there is $t' \in [\types]$ such that an agent of type $t'$ appears in $C_i$ and $C_i^{+(c,t)} \prec_{t'} C_i$,
			\end{itemize}
		where $C_i^{+(c,t)}$ is obtained from $C_i$ by adding one more agent of color $c$ and type $t$,
	\end{itemize}
	and for every $i,j \in [\ell],\, i \neq j$ we have
	\begin{itemize}
		\item $C_i \cap C_j = \emptyset$,
		\item $C_i$ contains no agent which admits NS/IS-deviation to $C_j$ and vice versa,
		\item if both $C_i$ and $C_j$ contain an agent of color $c$ and type $t$, then $C_i\succeq_t C_2^{c,t}$ or $C_j\succeq_t C_2^{c,t}$,
	\end{itemize}
	and at most $\coalnum$ coalitions are non-trivial and $\ell \le \scoalnum$.
	
	On the one hand a partial solution containing all agents is a solution.
	On the other hand, a solution is a partial solution containing all agents for the branch in which  $C_1^{c,t}$ and $C_2^{c,t}$ are set according to our semantics.
	Hence, our aim is to determine whether there is such a partial solution.
	
	Since the agents of the same type and color are interchangeable, we will slightly abuse the notation and
	call $C_1, \ldots, C_\ell$ a partial solution even if the coalitions are not disjoint, but there is enough agents of given types and colors such that the coalitions can be made disjoint by exchanging agents of the same type and color.
	
	Since we cannot store or even examine all partial solutions during the course of dynamic programming, we will collapse partial solutions that are in some sense equivalent into a single entry.
	To this end, we describe each partial solution with a \emph{pattern} formed by
	\begin{itemize}
		\item a mapping $a \colon [\colors]\times [\types]\to [n]_0$ representing for each $c \in [\colors]$ and $t \in [\types]$ the number of agents of color $c$ and type $t$ used by the partial solution,
		\item a mapping $w \colon [\colors]\times [\types]\to \{0,1\}$ representing for each $c \in [\colors]$ and $t \in [\types]$ whether the partial solution contains a coalition $C_i$ containing an agent of color $c$ and type $t$ such that $C_i \prec_t C_2^{c,t}$,
		\item and integers $r$ and $\ell$, representing the number of non-trivial and all coalitions in the partial solution, respectively.
	\end{itemize}
	
	Note that there are at most $(n+1)^{\colors \cdot \types+2}\cdot 2^{\colors \cdot \types}$ patterns.
	We call a pattern \emph{realized}, if there is a partial solution described by the pattern.
	Obviously, the \emph{zero pattern} formed by mappings $a$ and $w$ which are both identically zero and $r=\ell=0$ is realized by $\emptyset$.
	Our task is to determine whether there is a \emph{target} realized pattern $(a,w,r,\ell)$ such that $a(c,t)=n_{c,t}$ for each $c \in [\colors]$ and $t \in [\types]$, $r \le \coalnum$, and $\ell\le \scoalnum$ ($w$ can be arbitrary).
	
	We create a table storing for each pattern whether it can be realized or not.
	We initialize this table by marking each pattern as not realized and only the zero pattern as realized.
	Then we traverse the table in order of growing $\ell$ and try to add to (a partial solution described by) the realized pattern one more coalition. If this creates a partial solution, we reveal one more realized pattern. Let us describe this in more detail.
	
	Let $(a,w,r,\ell)$ be a realized pattern.
	Let $C \subseteq N$ be any non-empty coalition and let $r^C =0$ if $|C|=1$ and $r^C =1$ otherwise.
	Let $a^C \colon [\colors]\times [\types]\to [n]_0$ be a mapping such that for each $c \in [\colors]$ and $t \in [\types]$ the number of agents of color $c$ and type $t$ in $C$ is $a^C(c,t)$.
	Furthermore, let $w^C \colon [\colors]\times [\types]\to \{0,1\}$ be a mapping such that for each $c \in [\colors]$ and $t \in [\types]$ we have $w^C(c,t)=1$ if $a^C(c,t) \ge 1$ and $C \prec_t C_2^{c,t}$ and $w^C(c,t)=0$ otherwise.
	
	We call a coalition $C$ \emph{compatible} with pattern $(a,w,r,\ell)$ if for each $c \in [\colors]$ and $t \in [\types]$ the following conditions are met:
	\begin{itemize}
	 \item $|C| \le \coalsize$,
	 \item if $a^C(c,t) \ge 1$, then $C \succeq_t C_1^{c,t}$,
	 \item $a(c,t)+ a^C(c,t) \le n_{c,t}$,
	 \item $w(c,t)+ w^C(c,t) \le 1$,
	 \item $r + r^C \le \coalnum$,
	 \item $\ell + 1 \le \scoalnum$,
	 \item and
	 \begin{itemize}
	 \item either $w^C(c,t)=1$ and $C^{+(c,t)} \preceq_t C_2^{c,t}$,
	 \item or $w^C(c,t)=0$ and $C^{+(c,t)} \preceq_t C_1^{c,t}$,
	 \item or $a^C(c,t) = n_{c,t}$,
	 \item or alternatively, in case of IS, there is $t' \in [\types]$ with $\sum_{c' \in [\colors]} a^C(c',t') \ge 1$ such that $C^{+(c,t)} \prec_{t'} C$,
	 \end{itemize}
	\end{itemize}
	where $C^{+(c,t)}$ is obtained from $C$ by adding one more agent of color $c$ and type~$t$.
	
	As the algorithm traverses the table in order of growing $\ell$, whenever it reveals a realized pattern $(a,w,r,\ell)$, it cycles over all possible different coalitions (that is over all possible mappings $a^C$ determining such coalitions) and if $C$ is compatible with $(a,w,r,\ell)$, then it mark the pattern $(a',w',r',\ell')$ as realized, where $r'=r+r^C$, $\ell'=\ell+1$, and for each $c \in [\colors]$ and $t \in [\types]$ we have $a'(c,t)=a(c,t)+ a^C(c,t)$ and $w'(c,t)= w(c,t)+ w^C(c,t)$.
	
	Since there are at most $(n+1)^{\colors \cdot \types+2}\cdot 2^{\colors \cdot \types}$ patterns, $(n+1)^{\colors \cdot \types}$ possible  different coalitions, and the compatibility conditions can be checked in polynomial time, the DP algorithm works in $n^{\bigoh(\colors \cdot \types)}$ time.
	Furthermore, since there are $\big((n+1)^{2\colors^2 \cdot \types}\big)$ initial branches for $C_\star^{c,t}$'s, the total running time of the algorithm is $n^{\bigoh(\colors^2 \cdot \types)}$.
	
	It remains to prove that the algorithm fills the table correctly, that is, it marks a pattern realized if and only if it is indeed realized.
	We prove this by induction on $\ell$.
	The claim is trivial for $\ell=0$, as the zero pattern is the only realized pattern with $\ell=0$.
	Hence, let us assume that $k \ge 1$, the claim is true for all patterns with $\ell \le k$, and we want to prove it for patterns with $\ell=k+1$.
	
	We start with the only if direction. Suppose that a pattern $(a',w',r',\ell')$ with $\ell'=k+1$ was marked as realized while the algorithm was processing a pattern $(a,w,r,\ell)$ that was marked as realized and found a compatible coalition $C$ such that $r'=r+r^C$, $\ell'=\ell+1$, and for each $c \in [\colors]$ and $t \in [\types]$ we have $a'(c,t)=a(c,t)+ a^C(c,t)$ and $w'(c,t)= w(c,t)+ w^C(c,t)$.
	Since the pattern $(a,w,r,\ell)$ was marked as realized and the claim is true for it as $\ell\le k$, there is a partial solution $C_1, \ldots, C_\ell$ described by $(a,w,r,\ell)$. We claim that $(a',w',r',\ell')$ is realized by the partial solution $C_1, \ldots, C_\ell, C_{\ell+1}$, where $C_{\ell+1}=C$. Let us first check that it is indeed a partial solution. The conditions are obviously satisfied, unless $i=\ell+1$ (the case $j=\ell+1$ is symmetric). Since $C$ is compatible, we have $|C_i| \le \coalsize$ and for each $c \in [\colors]$ and $t \in [\types]$ we have $a(c,t)+ a^C(c,t) \le n_{c,t}$ implying that the coalitions can be made disjoint by exchanging agents of the same type and color. We also have $C \succeq C_1^{c,t}$ and if there was $j \in [\ell]$ such that $C_j \prec C_2^{c,t}$ and $C \prec C_2^{c,t}$, then $w(c,t)=1$ and $w^C(c,t)=1$, contradicting $w(c,t)+ w^C(c,t) \le 1$. Hence either $C_j \succeq C_2^{c,t}$ for every $j \in [\ell]$ or $C \succeq C_2^{c,t}$ and $C_1, \ldots, C_{\ell+1}$ is indeed a partial solution. It is straightforward to check, that it is described by pattern $(a',w',r',\ell')$.
	
	For the if direction, suppose that pattern $(a',w',r',\ell')$ is realized by a partial solution $C_1, \ldots, C_{\ell'}$.
	Let $(a,w,r,\ell)$ be the pattern describing the partial solution $C_1, \ldots, C_{\ell'-1}$ and $C=C_{\ell'}$.
	It follows from the definition that $r'=r+r^C$, $\ell'=\ell+1$, and for each $c \in [\colors]$ and $t \in [\types]$ we have $a'(c,t)=a(c,t)+ a^C(c,t)$ and $w'(c,t)= w(c,t)+ w^C(c,t)$.
	Since $\ell \le k$, the pattern $(a,w,r,\ell)$ is marked as realized by the induction hypothesis.
	It is straightforward to check that $C$ is compatible with $(a,w,r,\ell)$.
	Hence the algorithm marks $(a',w',r',\ell')$ as realized, finishing the proof.
\end{proof}

The final result in this section is an \XP-algorithm for \HDGshort parameterized by the number of colors and non-trivial coalitions by combining branching with a flow subroutine similar to one used by Darmann et al.~\cite{DarmannEKLSW12}.

\begin{theorem}\label{lem:allHDGNash:XP:ColorsWeakNum}
	\allHDGNashShort and \allHDGIndividualShort are in \XP parameterized by $\colors+\coalnum$.
\end{theorem}
\begin{proof}
	First, consider the case of Nash stability. We begin by branching to determine the number of non-trivial coalitions (at most $\coalnum$ many), and for each non-trivial coalition and each color we branch to determine the number of agents of that color in the coalition; this, in total, yields a branching factor of at most $\coalnum\cdot (n+1)^{\coalnum\cdot\colors}$. We will view each branch as a ``guess'' of these properties of a hypothetical targeted Nash-stable outcome.
	
	If any of the coalitions is of size more than $\coalsize$ we discard the guess.
	For each color we sum up the number of agents of that color in the non-trivial coalitions and if this is more than the number of agents of that color available in the instance, we discard the guess.
	Otherwise we denote by $n_c$ the number of agents of color $c$ that should be in a trivial coalitions.
	Hence, there will be $\sum_{c \in [\colors]}n_c$ trivial coalitions. If this number together with the number of non-trivial coalitions is more than $\scoalnum$, we discard the guess.
	
	Our task is now to determine where each agent $a\in N$ can be placed. To this end, we build an auxiliary bipartite graph whose vertices are agents on the one side and coalition-color pairs on the other side.
	Let $\mathcal{C}$ be the set corresponding to our guessed coalitions (both trivial and non-trivial ones).
	We say that $C \in \mathcal{C}$ is \emph{valid} for an agent $a \in N$ of color $c$ iff
	\begin{itemize}
		\item
	     $c$ should appear in the coalition corresponding to $C$ and
	    \item
	     $a$ weakly prefers the coalition corresponding to $C$ over a coalition arising from adding \(a\) to any coalition corresponding to some \mbox{$\hat{C} \in \mathcal{C}$} or the empty coalition; (all the information required to determine this is the color composition of the coalition corresponding to \(C\) which we have already branched on).
	\end{itemize}
	
	Now, we can build the auxiliary graph.
	On the left side there is the agent set $N$ and on the right side there is the set of pairs $(C,c)$ for $C \in \mathcal{C}$ and a color $c$ that is present in~$C$.
	There is an edge between $a \in N$ and $(C,c)$ if $a$ has color $c$ and $C$ is valid for $a$; all these edges have capacity~$1$.
	We add a source~$s$ and connect it by an edge of capacity~$1$ with all agent vertices.
	We add a sink vertex~$t$ and connect it with all vertices in the right side (i.e., vertices of the form $(C,c)$) and set  their capacity to the guessed number of agents of color~$c$ in the coalition~$C$.
	See \Cref{fig:flow} for an illustration of the created network.
	\begin{figure}
		\centering
		\includegraphics[scale=0.7]{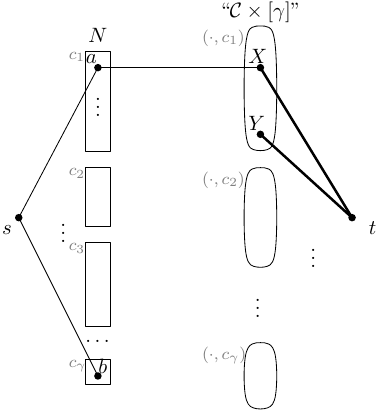}
		\caption{Illustration of the network flow instance used in the proof of \Cref{lem:allHDG:XP:WeakNumSize}\label{fig:flow}.
		Agents' membership in sets of a certain color set are indicated by rectangles; correspondingly the coalition color tuple's second entries are indicated by the rounded rectangles.
		All except the thick edges have unit capacity, and the capacity on the thick edges is determined by the branching.
		The existence of edges in the middle is determined by the branching as well.
		There will never be edges between agents of a certain color and coalition color tuple with a different second entry, e.g. there will never be an edge between \(b\) and \(X\) or \(Y\).
		In this example, the coalition that is the first entry of \(X\) is valid for \(a\), while the one of \(Y\) is not.}
	\end{figure}
	We verify whether it is possible to assign all the agents to appropriate coalitions using any max flow algorithm in polynomial time; see, e.g., \cite{EdmondsK72}.
	Due to the construction, each agent is assigned to an element of \(\mathcal{C}\) that is valid for her. This yields a Nash stable outcome for the given instance.
	The total running time is $(n+1)^{\coalnum \cdot \colors +\bigoh(1)}$.
	
	\smallskip
	Let us now turn to the more complicated case of individual stability, where we have to further refine the information we obtained by guessing.
	We additionally guess for each coalition and each color, whether an agent of this colors would be accepted in the coalition by the agents already in the coalition.
	Since the trivial coalitions are similar in that regard, we merely guess for each pair of colors $c,c'$ whether there will be a trivial coalition of color $c$ accepting agents of color $c'$. If there should be none, then we call $(c,c')$ a \emph{blocked} pair. Note that an agent in a trivial coalition is always indifferent about agents of the same color joining, so we exclude pairs formed by the same color.
	
	Furthermore, for each non-trivial coalition corresponding to some $C \in \mathcal{C}$ and for a color $c \in [\colors]$ such that no more agents of color $c$ should be accepted in the coalition corresponding to $C$,
	we guess the \emph{blocking} agent, that is, an agent $i$ such that $C^{+c} \prec_i C$, where $C^{+c}$ is obtained from the coalition corresponding to $C$ by adding one more agent of color $c$.
	Agent $i$ must have a color that should appear in $C$. The same agent can be reused for several colors, however not for different coalitions.
	This makes at most $(n+1)^{\coalnum \cdot \colors}$ guesses for the non-trivial coalitions and $2^{\colors^2}$ guesses for the trivial coalitions.
	
	For $c \in [\colors]$ let $\mathcal{C}_c$ be the set of guessed coalitions designated to accept agents of color $c$ according to our guess (including trivial ones).
	We say that a coalition $C \in \mathcal{C}$ is \emph{valid} for an agent $a \in N$ of color $c$ if
	\begin{itemize}
		\item color $c$ should appear in the coalition corresponding to $C$ and
		\item $a$ weakly prefers the coalition corresponding to $C$ over a coalition arising from adding \(a\) to any coalition corresponding to some \mbox{$\hat{C} \in \mathcal{C}$} or the empty coalition; (all the information required to determine this the color composition of the coalition corresponding to \(C\) which we have already branched on) and
		\item if $C$ is trivial, then for every blocked pair $(c,c')$, we must have that \(c\) strictly prefers the coalition corresponding to $C$ to \(C^{+c'}\), where $C^{+c'}$ is obtained from the coalition corresponding to $C$ by adding an agent of color $c'$.
	\end{itemize}
	
	We again build the graph. A blocking agent can be only assigned to the coalition it is blocking, and only if it is valid for the coalition. If it is not valid for the coalition it should be blocking, then we may discard the guess. For the agents not blocking any coalition, we add the edges as before.
	The rest of the algorithm stays the same. The total running time is $2^{\colors^2}(n+1)^{2\coalnum \cdot \colors +\bigoh(1)}$.
\end{proof}

Observant readers will notice that there is one more tractability result in \Cref{fig:complexityPicture}, notably an \XP-algorithm parameterized by $\colors+\scoalnum$. Since we may assume that, w.l.o.g., $\scoalnum\geq \coalnum$, this follows as a corollary of \Cref{lem:allHDGNash:XP:ColorsWeakNum}.

\begin{corollary}\label{lem:allHDGNash:XP:ColorsStrongNum}
	\allHDGNashShort and \allHDGIndividualShort are in \XP parameterized by $\colors+\scoalnum$.
\end{corollary}

\section{Algorithmic Lower Bounds}

This section complements our algorithms and provides all the remaining results required to obtain a complexity picture for \HDGshort. We begin by recalling that \allHDGNashShort\ remains \NP-complete even when $\colors=2$ and \allHDGIndividualShort is \NPc even when $\colors=5$~\cite{BoehmerE2020}. The first two results in this section establish the remaining \paraNP-complete cases.

Before proceeding to the reductions, we build a general trap gadget that is similar in spirit to~\cite[Example~5.3]{BoehmerE2020} and may be of independent interest. The gadget uses agents of three colors---red, blue, and green; we assume these colors are distinct from the colors used outside this gadget.
  Let $R$ denote the red agent, $B$ the blue agent, and $G$ the green agents of which there are $m$ in total; here $G^+R$ is the profile $G^mR \succ \cdots \succ GGR \succ GR$, and similarly for $G^+B$ and $G^+$.
We will use $\mathcal{M}_G$ to denote all ``desirable'' coalitions containing at least one agent outside of the gadget, i.e., one agent that is neither red nor blue nor green.
\begin{align*}
    R &:: \textstyle RB \succ G^+R \succ R \succ \cdots  \\
    B &:: \textstyle G^+B \succ RB \succ B \succ \cdots  \\
\forall g\in G &:: \textstyle \mathcal{M}_G \succ G^+R \succ G^+B \succ G^+ \succ \cdots,
\end{align*}
Note that one can design this gadget using palettes if and only if one can design~$\mathcal{M}_G$ in this way.
Furthermore, observe that this can be done using $3$ additional master lists.

\begin{claim}\label{clm:greensMustBeInMG}
	If there is a green agent who is not part of a coalition in~$\mathcal{M}_G$, then the outcome is neither IS nor NS.
\end{claim}
\begin{claimproof}
	Suppose there is a green agent that does not participate in $\mathcal{M}_G$ and suppose towards the contradiction that~$\Pi$ is stable.
	Let us call the green agents not participating in $\mathcal{M}_G$ as the \emph{free agents}.
	We begin with observation that in~$\Pi$ it is impossible for the $R$ and $G$ agents to be all alone.
	
	Suppose $\{R,B\} \in \Pi$.
	Then, it must be the case that all free agents have to be partitioned into purely green coalitions.
	This is not individually stable as the blue agent want to join a coalition of free agents which would accept it; this is an individual deviation. Let $F$ be a non-empty subset of free-agents.
	If we have $\{ B \} \cup F \in \Pi$ (and consequently $\{ R \} \in \Pi$), then this is not individually stable, since the green agents want to join the red agent (who will accept them).
	If $\{ R \} \cup F \in \Pi$, then the blue agent~$B$ is alone and thus the red agent wants to join them (and is accepted).
	However, this yields the original situation; we have proven that $\Pi$ is not Nash/individually stable.
\end{claimproof}

\begin{theorem}\label{lem:allHDG:NPh:SizeTypes}
	\allHDGNashShort and \allHDGIndividualShort are \NPh even when $\coalsize = \types = 4$.
\end{theorem}
\begin{proof}
	We present a reduction from a classical \NPc problem~\cite{GareyJ1979} called \XThreeC (\XThreeCshort): given a set $U=\{1,\ldots,3m\}$ and a family $\mathcal{X}\subseteq \binom{U}{3}$, does there exist a set $S\subseteq \mathcal{X}$ such that every element of $U$ occurs in exactly one element of $S$?

	Let $\mathcal{I} = (U,\mathcal{X})$ be an instance of \XThreeCshort with $|U| = 3m$.
	We introduce an agent $a_u$ for every $u \in U$ with the (new) color $u$; we call these \emph{universe agents}.
	We want these agents to stay in a coalition in $\mathcal{X}$ together with one green agent, i.e., in a coalitions of the following kind
	\(
	 \left\{ X \cup \{G\} \mid X \in \mathcal{X} \right\}
	\), where $G$ is a green agent.
	For these agents we present a single master list: $\left\{ X \cup \{G\} \mid X \in \mathcal{X} \right\} \succ U$.
	That is, the agents can either be part of a desired coalition, or alone. We now apply the trap gadget construction described above to prevent stable outcomes where agents are alone---in particular, we will use $\mathcal{M}_G = \left\{ X \cup \{G\} \mid X \in \mathcal{X} \right\}$.
	
	In the rest of the proof we verify that indeed the two instances are equivalent.
	We do this for both Nash and individual stability at once.
	
	If $\mathcal{I}$ is a \YES-instance with solution $S_1, \ldots, S_m$, it is not hard to verify that $\Pi = \left( S_i \cup \{G\} \right)_{i \in [m]} \cup \{R,B\}$ is a Nash-stable outcome.
	
	In the other direction let us assume that $\Pi$ is a stable outcome.
	By \Cref{clm:greensMustBeInMG} we get that there are exactly $m$ coalitions using a green agent and three universe agents; note that the three universe agents must form a set in $\mathcal{X}$ as otherwise they would rather be alone.
	Consequently , we have found a cover of $U$ using exactly $m$ elements from $\mathcal{X}$ and thus $\mathcal{I}$ is a \YES-instance.
\end{proof}

The proof for our second intractable case is based on a reduction from the \NP-complete \PTTN problem~\cite{GareyJ1979}: given a multiset $\SSS$ of integers, is there a partition of $\SSS$ into two distinct subsets $\SSS_1$ and $\SSS_2$ such that the sum of elements in $\SSS_1$ equals the sum of numbers in $\SSS_2$?

\begin{theorem}\label{lem:allHDG:NPh:TypesNum}
	\allHDGNashShort is \NPh even when restricted to instances where ${\types=\coalnum = 2}$ and ${\scoalnum=3}$.
\end{theorem}
\begin{proof}
	Let $\III$ be an instance of \PTTN\ where the sum of all numbers is $2k$. We construct an equivalent instance of \allHDGNashShort as follows. For every number $i\in\SSS$ we introduce a \emph{normal} agent, distinguishing multiplicities of the same number. Every such agent is assigned a unique color, which allows us to distinguish them. We present a single master list for every agent such that on the first place there are all coalitions such that the colors of their members corresponds a choice of numbers that sum to $k$.
	Notice that this list is very large, but can be encoded using an oracle with polynomial running time. As their second choice, all the agents want to be in coalitions consisting entirely of their own color, which can be done continuing the single master list, as coalitions consisting of a color are impossible to form with agents that are not of this color.

	To complete the construction for \allHDGNashShort, we add into the instance one green (guard) agent $G$ such that their top choice is to be in a coalition consisting entirely of one other color apart from green or in the grand coalition. The second choice is to form a trivial coalition.

	Let $\III$ be a \YES-instance of the \PTTN problem and $\SSS_1,\SSS_2$ the desired partitioning. If we create coalitions $C_1 = \SSS_1$, $C_2 = \SSS_2$ and $C_3 = \{G\}$, we obtain a Nash stable and individually stable outcome, since all the agents are in the coalition of their top choice with the exception of $G$, who is alone but no preferable coalition exists for them.

	On the other hand, let $\III'$ be a \YES-instance of \allHDGNashShort with $\types = \coalnum = 2$ and $\scoalnum = 3$ and agents as prescribed above. Due to agent $G$ there are at least $2$ coalitions, since the grand coalition is not stable for normal agents. For a stable solution consisting of at least $2$ coalitions we know that one of them is the trivial coalition with agent $G$, as normal agents would want to deviate from coalitions with \(G\). Any second non-trivial coalition contains only normal agents and by their preferences must correspond numbers that sum up to \(k\). Hence we know that $G$ is in a trivial coalition, say $C_3$, and the two remaining coalitions form exactly the wanted partitions $\SSS_1$ and $\SSS_2$.
\end{proof}

Note that \allHDGNashShort and \allHDGIndividualShort are trivially solvable in polynomial time in case of $\scoalnum = 1$. The only possible coalition is the grand coalition and we can check for stability in polynomial time.

We remark that the instances produced by the reductions presented in the proofs of \Cref{lem:allHDG:NPh:SizeTypes,lem:allHDG:NPh:TypesNum} never admit any stable outcome that violates the restrictions imposed by the respective values of $\coalsize$, $\coalnum$ and $\scoalnum$ in the reductions.
More explicitly, in the instances constructed in the proof of \Cref{lem:allHDG:NPh:SizeTypes} there is no stable outcome in which any coalition has size more than four, and in the instances constructed in the proof of \Cref{lem:allHDG:NPh:TypesNum} there is no stable outcome with more than three coalitions in total and at most two non-singleton coalitions.
This implies that the hardnesses also holds when we consider\allHDGNashShort without restrictions on $\coalsize$, $\coalnum$ and $\scoalnum$ (i.e.\ setting each of these upper-bounds to \(n\)).

At this point, we know that none of our \XP\ algorithms can be strengthened by dropping a parameter in the parameterization. The remaining task for this section is to show the same for fixed-parameter algorithms, notably via matching \W[1]-hardness reductions.

For our first proof, we provide a reduction from a variant of a problem called \MSS (\MSSshort).
In the original formulation of \MSSshort, the input consists of integers $k, \omega\in\Nat$, a set $S=\{s_1,\ldots,s_n\}$ of $k$-dimensional vectors of non-negative integers and a target vector $t\in\Nat^k$.
The question is whether there exists a subset $S'\subseteq S$ of cardinality at most $\omega$ such that $\sum_{s\in S'}s = t$.

\MSSshort is known to be \Wh when parameterized by $k$~\cite[Lemma 3]{GanianKO2021}, and the same is true for several of its variants~\cite{GanianOR2019,GanianOR21}; crucially, all of these lower bound results hold even if every number on the input is encoded in unary. For our purposes, it will be useful to reduce from a partitioned variant of \MSSshort: instead of a set $S$, the input contains sets $S_1,\dots, S_\omega$ of $k$-dimensional vectors of non-negative integers, and the question is whether there exists a tuple $S'=(s_1,\dots,s_\omega)$, $\forall i\in [\omega]: s_i\in S_i\cup \{\bar{0}\}$ where \(\bar{0}\) denotes the \(k\)-dimensional vector with only zeros as entries, such that the vectors in $S'$ sum up to $t$. While this was not stated explicitly, the original proof of Ganian et al.~\cite{GanianKO2021} immediately establishes the \W[1]-hardness of partitioned \MSSshort parameterized by $\omega+k$: their reduction from \textsc{Multicolored Clique} parameterized by the clique size $\kappa$ results in instances of \MSSshort where vectors are partitioned into $\kappa+{\kappa\choose 2}$ sets and a solution can choose at most one element from each set (in fact, their construction ensures that exactly one is chosen per set).

Our first lower bound excluding fixed-parameter tractability targets the parameterization by the number of colors, number of types, and the number of coalitions.

\begin{theorem}\label{lem:allHDG:Wh:ColorsNumTypes}
	\allHDGNashShort and \allHDGIndividualShort are \Wh parameterized by $\colors+\scoalnum+\types$.
\end{theorem}
\begin{proof}
	Given an instance $\III=(S_1,\dots,S_\omega,t,k)$ of partitioned \MSSshort, our reduction constructs an instance of \HDGshort\ as follows. For each $i\in [\omega]$, we introduce one \emph{marker} agent with color $m_i$ (a \emph{marker} color) and for each $j\in [k]$ we introduce $t[j]$-many (where $t[j]$ is the value of $t$ on the $j$-th coordinate) normal agents, all with color $g_j$; collectively, we refer to all the former agents as \emph{markers} and all the latter agents as \emph{normal}. Intuitively, the $\omega$ marker agents will each model the selection of one vector in $S'$ and the normal agents represent the numbers in individual coordinates (which must sum up to $t$). The total number of colors is $k+\omega$ and we set\footnote{We remark that the proof can be extended to any value of $\scoalnum>\omega$ by using the trap gadget described above \Cref{clm:greensMustBeInMG}.} $\scoalnum$ to $\omega$.
	
	Formally, the preferences of the agents are as follows. Marker agents have trichotomous preferences and each have their own agent type. In particular, a coalition $C$ is \emph{satisfying} for the marker agent with color $m_i$, $i\in [\omega]$ if: (1) $C$ contains no other marker color than $m_i$ and (2) there exists $s\in S_i$ such that for each $j\in [k]$, the number of agents of color $g_j$ in $C$ is precisely $s[j]$. Condition (2) is formalized via setting the palette to $\left(\frac{s[j]}{1+\sum_{j'\in [k]}s[j']}\right)$ for each color $g_j$. Marker agents strictly prefer being in a satisfying coalition than being alone (i.e., in a coalition containing only a single color $m_i$), and strictly prefer being alone over any other non-satisfying coalition.
	
	Normal agents all have the same agent type and also have trichotomous preferences: they strictly prefer being in a coalition containing precisely only one marker color $m_i$ over being in a coalition containing no marker color, and strictly prefer being in a coalition containing no marker color over being in a coalition containing more than one marker color. This concludes the description of our instance $\III'$ of \HDGshort. Observe that $\tau=1+\omega$. The construction can be carried out in polynomial time.
	For an exemplary illustration of the construction see \Cref{fig:MSSreduction}.
	\begin{figure}
		\centering
		\includegraphics[scale=0.7]{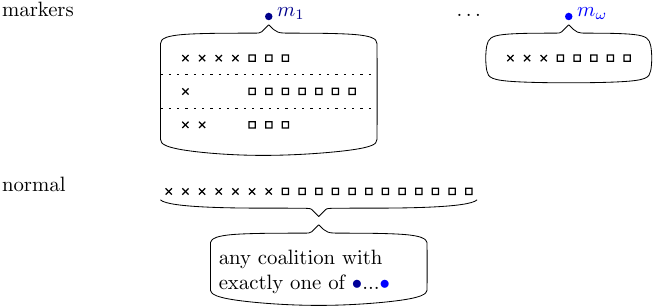}
		\caption{Exemplary illustration of the constructed \HDGshort\ instances in the proof of \Cref{lem:allHDG:Wh:ColorsNumTypes} with \(k = 2\).
		The textbubbles give the coalitions which the respective agents would prefer to be included in rather than being alone.
		Each marker (colored with one of \(\omega\) marker colors, indicated as different shades of blue) is associated to one of \(S_1, \dotsc, S_\omega\) corresponding to its preferred coalitions; in this example \(S_1\) consists of three vectors \((4,3),(1,7),(2,3)\), and \(S_\omega\) consists of only one vector \((3,5)\).
		The normal agents encode \(t\); in this example \(t = (7,12)\).\label{fig:MSSreduction}}
	\end{figure}
	
	To prove correctness, assume $\III$ admits a solution $S'=(s_1,\dots, s_\omega)$. Then $\III'$ is also a \YES-instance: for each $s_i$, $i\in [\omega]$ we construct a coalition containing the marker agent with color $m_i$ and for each $j\in [k]$ precisely $s_i[j]$-many normal agents of color $g_j$. All marker and normal agents are in their most preferred kind of coalitions, and hence the outcome is both Nash and individually stable.
	
	For the converse direction, assume that $\III'$ is a \YES-instance and in particular admits a stable outcome $\Pi = \{C_1,\ldots,C_\ell\}$ for some $\ell\leq \omega$. First of all, observe that since marker agents prefer being alone to sharing a coalition with another marker agent, it must hold that $\ell=\omega$ and each coalition in $\Pi$ contains precisely one marker agent; w.l.o.g., let us assume that each coalition $C_i$, $i\in [\omega]$, contains marker agent $m_i$. Since there are no other coalitions in $\Pi$, all normal agents must occur in coalitions which each contain precisely one marker agent; each such coalition $C_i$ containing a marker agent and at least one normal agent must be satisfying, since otherwise the marker agent would prefer being alone and thus create an NS- and IS-deviation. Hence $\Pi$ may contain only two kinds of coalitions: those which are satisfying, and those containing a single marker agent. For each satisfying coalition $C_i$, there must exist a vector $s_i\in S_i$ such that for each $j\in [k]$ there are precisely $s_i[j]$ many normal agents of color $g_j$ in $C_i$; if $C_i$ contains only a single agent and is hence not satisfying, we set $s_i$ to empty. Now consider the tuple $S'=(s_1,\dots,s_\omega)$, and observe that since all normal agents have been assigned to the coalitions in $\Pi$, it must hold that $\sum_{s\in S'}s = t$. Hence $S'$ is a solution for $\III$, concluding the proof.
\end{proof}

The last lower bound required to complete our complexity picture is for the case parameterized by the size of coalitions, number of types, and the number of non-trivial coalitions.

\begin{theorem}\label{lem:allHDG:Wh:SizeTypeNum}
	\allHDGNashShort and \allHDGIndividualShort are \Wh parameterized by $\coalsize+\types+\coalnum$.
\end{theorem}
\begin{proof}
	We design a parameterized reduction from the \textsc{Independent Set} problem parameterized by the size of the solution.
	Let $(G = (V,E), k)$ be an instance of \textsc{Independent Set}.
	There is a vertex agent $a_v$ for every $v \in V$ with the (new) color $v$.
	There is a green (guard) agent $G$ and a red and a blue agents that together constitute the trap gadget we built in the proof of \Cref{lem:allHDG:NPh:SizeTypes}.
	The approved coalitions $\mathcal{M}_G$ of the guard agent are all coalitions of size exactly $k+1$ where this agent is in a coalition with any $k$ vertex agents.
	Note that this is expressible using the ratios as there is a unique guard agent.
	The vertex agents are indifferent between any coalition of size $k+1$ that contain the guard agent and do not contain any edge; if this is not possible; they only want to be alone.
	It is not hard to verify that this is a single master list, since vertex agents have unique colors.
	We note that these preference lists are large -- in fact, they are as large as $\binom{|V|}{k}$ but we can encode them using a very simple oracle (with polynomial running time).
	We set $\coalnum=1$ and observe that $\types\leq 4$; this finishes the description of the instance of \HDGshort\ and hence of the reduction.

	For the correctness, let us first assume there is no independent set of size (at least) $k$ in the graph~$G$.
	Then, the vertex agents do only approve singleton coalitions.
	It follows that the guard agent cannot be part of any coalition of size $k+1$ and thus such an instance does not admit any Nash/individually stable outcome.
	In the opposite direction, if $G$ contains an independent set of size at least $k$, then there are some coalitions of size $k+1$ approved by vertex agents and thus these together with the guard agent form the only non-trivial coalition---forming a stable outcome.
\end{proof}

\section{Parameterizing by the Smaller of Two Color Classes}

Boehmer and Elkind~\cite{BoehmerE2020} showed that the problem of finding a Nash stable outcome for hedonic diversity game with $2$ colors $R$, $B$ is in \XP when parameterized by $q=\min\{|R|,|B|\}$. Immediately after proving that result, they ask whether this problem is in \FPT with respect to the same parameter. We answer this in the negative with a \W-hardness reduction.

For our reduction we start from a restricted version of the so called \textsc{Group Activity Selection} problem~\cite{DarmannDDLS17}, called the \textsc{Simple Group Activity Selection Problem} (\textsc{sGASP}).
To distinguish terminology between the \textsc{Group Activity Selection} problem and our HDG setting, we use slightly different terminology for that problem.

We are given a set of \emph{participants} \(P\), a set of \emph{activities} \(A\) and for each participant \(p \in P\) a set of \emph{approved group sizes} \(S_p \subseteq \{(a,i) \mid a \in A, i \in [|P|]\}\) for the activities.
The task is to decide whether there is an assignment \(\pi\) of participants to activities which is stable, i.e., whether there is \(\pi\) such that for every participant \(p \in P\), \((\pi(p), |\pi^{-1}(\pi(p))|) \in S_p\).
\textsc{sGASP} is known to be \Wh[1] parameterized by the number of activities~\cite{EibenGO18}.

For our purposes it will be useful to additionally have only odd approval sets and solutions in which the number of participants assigned to different activities do not differ by large factors.

\begin{lemma}\label{lem:w1:gasp}
	\textsc{sGASP} remains \Wh parameterized by the number of activities \(|A|\), even restricted to instances for which there is \(s \in \mathbb{N}\) such that for every participant \(p\) and \((a,t) \in S_p\), \(t\) is an odd number in \(\{2s|A|+1,\ldots, 2s|A| +2s-1\}\) and, moreover, each activity is approved by at most \(2s|A| +2s-1\) participants.
\end{lemma}
\begin{proof}
	Let \(\III\) be an instance of \textsc{sGASP} with participants \(P\), activities \(A\) and approval sets \((S_p)_{p \in P}\), let $s=|P|+1$.

	We create the following equivalent instance \(\tilde{\III}\) with participants \(\tilde{P}\), activities \(\tilde{A} = A\) and approval sets \((\tilde{S}_p)_{p \in \tilde{P}}\) of \textsc{sGASP} which satisfies the properties of the lemma statement.

	Let \(\tilde{P}\) be given by a set of \((2s|A| + 1)|A|\) `activity participants', \(2s|A| + 1\) associated to each activity; and \(2|P|\) `pair participants' two of which are associated to each participant of \(\III\).
	Now for each activity participant \(p\) associated to activity \(a\), let \(\tilde{S}_p\) be given by all tuples \((a,t)\) where \(t\) is an odd number between \(2s|A| + 1\) and \(2s|A| +2s - 1\).
	For each pair participant \(p\) associated to the original participant \(p' \in P\), let \(\tilde{S}_p\) be given by \(\{(a,2t + 2s|A| + 1) \mid (a,t) \in S_{p'}\}\).
	Observe that by construction, for every \(p \in \tilde{P}\) and \((a,t) \in \tilde{S}_p\), \(t\) is odd in \(\{2|A|s+1,\ldots, 2s|A| +2s-1\}\) and each activity is approved by at most \(2s|A| +2s-1\) participants.

	We claim that \(\III\) is a \YES-instance if and only if \(\tilde{\III}\) is a \YES-instance.

	For one direction assume that \(\pi\) is a solution for \(\III\).
	Then it is easy to check that a solution for \(\tilde{\III}\) is given by mapping every activity participant to the activity it is associated to, and every pair participant associated with the participant \(p \in P\) to \(\pi(p)\).

	Conversely assume that there is a solution \(\tilde{\pi}\) for \(\tilde{\III}\).
	Assume there is some pair participant \(p_1\) associated to \(p \in P\) such that the second pair participant \(p_2\) associated to \(p\) is assigned to a different activity by \(\tilde{\pi}\), i.e., \(\tilde{\pi}(p_1) \neq \tilde{\pi}(p_2)\). The size \(|\tilde{\pi}^{-1}(\tilde{\pi}(p_1))|\) must be odd by the choice of the approved group sizes, and must contain all activity participants associated to \(\tilde{\pi}(p_1)\) by the choice of the approved group sizes of these activity participants.
	Hence there must be at least one pair participant \(p'_1 \in \tilde{\pi}^{-1}(\tilde{\pi}(p_1))\) associated to \(p' \neq p\) such that \(p'_2\) also associated to \(p'\) is not mapped to \(\tilde{\pi}(p_1) = \tilde{\pi}(p'_1)\) by \(\tilde{\pi}\).
	In this way we can find a sequence \(p^{(1)}, p^{(2)}, \dotsc, p^{(\ell)}, p^{(\ell + 1)}\) where \(\ell\) is an even number, with \(p^{(1)} = p^{(\ell + 1)} = p_2\), \(p^{(2)} = p_1\), \(p^{(3)} = p'_1\) and \(p^{(4)} = p'_2\)
	such that
	\begin{itemize}
		\item for all odd \(i \in [\ell]\), \(p^{(i)}\) and \(p^{(i + 1)}\) are associated to the same participant in \(\III\) and are mapped to different activities by \(\tilde{\pi}\); and
		\item for all even \(i \in [\ell]\), \(p^{(i)}\) and \(p^{(i + 1)}\) are mapped to the same activity by \(\tilde{\pi}\) and are associated to different participants in \(\III\).
	\end{itemize}
	We can modify \(\tilde{\pi}\) to map \(p^{(i + 1)}\) to \(\tilde{\pi}(p^{(i)})\) for all odd \(i \in [\ell]\).
	More formally, consider \[\widehat{\pi}(p) = \begin{cases}
		\tilde{\pi}(p^{(i)}) & \mbox{if } p = p^{(i + 1)} \mbox{ for odd } i \in [\ell]\\
		\tilde{\pi}(p) & \mbox{otherwise.}
	\end{cases}\]
	Then for all \(a \in A\), \(|\tilde{\pi}^{-1}(a)| = |\widehat{\pi}^{-1}(a)|\), and the fact that for all odd \(i \in [\ell]\), \(p^{(i)}\) and \(p^{(i + 1)}\) have the same approval set implies that \(\widehat{\pi}\) is also a solution for \(\III\).

	We can iterate this procedure and assume that all pair participants which are associated to the same participant in \(\III\) are assigned to the same activity by \(\tilde{\pi}\).

	Now we observe that \(\pi\) which maps each \(p \in P\) to the same activity as \(\tilde{\pi}\) maps both the associated pair participants is a solution for \(\III\):
	In this way for every \(p \in P\) \(|\pi^{-1}(\pi(p))| = \frac{|\tilde{\pi}^{-1}(\tilde{\pi}(p'))| - 2s|A| - 1}{2}\) where \(p' \in \tilde{P}\) is a participant in \(\tilde{\III}\) associated to \(p\).
	By definition of \(\tilde{S}_{p'}\) this implies that \((\pi(p), |\pi^{-1}(\pi(p))|) \in S_p\).
\end{proof}

\begin{theorem}\label{thm:w1}
	\allHDGNashShort restricted to instances with $\gamma=2$ is \Wh when parameterized by the size $q$ of the smaller color class.
\end{theorem}
\begin{proof}
	We reduce from \textsc{sGASP} restricted to instances as described in \Cref{lem:w1:gasp} parameterized by the number of activities.
	Since the instances with $|A| \le 2$ can be solved in polynomial time, we assume $|A|\ge 3$.

	Let \(\III\) be such an instance of \textsc{sGASP} with participants \(P\), activities \(A\) and approved group sizes \((S_p)_{p \in P}\).
	Moreover we consider an auxiliary sequence of
	integers \((z_i)_{i \in [|A|]}\) where simply \(z_i = 100i + 1\).

	We construct an instance \(\JJJ\) of \HDGshort with the following agents:
	\begin{itemize}
		\item \(|P|\) `normal agents' which all have color \emph{blue} and each of which corresponds to one participant in \(\III\).
		\item \(\sum_{i \in [|A|]} z_i\) `marker agents' of color \emph{red} such that for each \(i \in [|A|]\), \(z_i\) of these marker agents have the same preferences and correspond to the \(i\)-th activity.
		\item \((400|A|^2)200|A|^2 + 1\) `spoiler agents' which will have the same preferences and all have color \emph{red}.
	\end{itemize}
	Note that in this way the number of red agents is polynomial in \(|A|\), the parameter of the \W[1]-hard problem \textsc{sGASP}.

	Let the agents' preferences be given as follows:
	\begin{itemize}
		\item For each normal agent \(j\) corresponding to a participant \(p\) in \(\III\), we let red ratios in coalitions be preferred in the order
		\[
		\Big\{\frac{z_i}{z_i + t} \,\Big\vert\, (a,t) \in S_p \land a \text{ is the } i\text{-th activity}\Big\}  \succ_j 0 \succ_j \dots
		\]
		\textbf{Intuition:}
		The number of red agents in a coalition encodes which activity the participants corresponding to the blue agents in that coalition should be assigned to.
		Then the preferences of the normal agents ensure that at least in terms of ratio, the resulting group sizes behave like the approved group sizes would.
		\item For each marker agent \(j\) corresponding to the \(i\)-th activity in \(\III\), we let red ratios in coalitions be preferred in the order
		\[\Big\{\frac{z_i}{z_i + 2s(|A| + 1)-1}, \ldots, \frac{z_i}{z_i + 2s|A|+1}\Big\} \succ_j 1 \succ_j \dots.\]
		\textbf{Intuition:}
		For the preferences of the normal agents to take into account \emph{all} participants corresponding to blue agents which will be assigned to the same activity we have to ensure that all the blue agents that will be assigned to the same activity will be in the same coalition.
		The marker agents should identify which unique coalition will correspond to which activity.
		Correspondingly their preferences allow to be in the same coalition as the marker agents associated with the same activity and an arbitrary number of normal agents between \(2s|A|+1\) and \(2s(|A| + 1)-1\).
		\item For each spoiler agent \(j\), we let red ratios in coalitions be preferred in the order \(\mathfrak{B} \succ_j \mathfrak{S} \setminus \mathfrak{B} \succ_j 1 \succ_j \dots\),
		where
		\[\mathfrak{B} = \Big\{\frac{1}{1+b} \mid b \in \mathbb{N}\Big\}
		\]
		is the set of `blue-up-to-1' ratios and
		\begin{align*}
			\mathfrak{S} = \Big\{\frac{r + 1}{r + b + 1} \,\Big\vert\,& \exists i \in [|A|],\\
			& \exists t \in \{2s|A|+1,\ldots 2s(|A| + 1)-1\}\\
			& \mbox{odd,} \ \frac{r}{r + b} = \frac{z_i}{z_i + t} \land r \leq 75i + 1\Big\}
		\end{align*}
		is the set of `small-split' ratios.

		\textbf{Intuition:}
		The spoiler agents can in some sense be regarded as the crux of our reduction.
		They will ensure that (1) there is no stable outcome in which there is a coalition with only blue agents and (2) in every stable outcome there is exactly one coalition corresponding to each activity and that the marker agents for one activity are not distributed in many coalitions.
	\end{itemize}

	\begin{claim}
		If \(\III\) is a \YES-instance, then \(\JJJ\) is also a \YES-instance.
	\end{claim}
	\begin{claimproof}
		Let \(\pi\) be a solution for \(\III\).
		Consider \(\Pi = \{C_1, \dotsc, C_{|A| + 1}\}\) with the following coalitions.
		For \(i \in [|A|]\), coalition \(C_i\) is comprised of the \(z_i\) marker agents corresponding to activity \(i\) together with the normal agents corresponding to \(\pi^{-1}(i)\).
		\(C_{|A| + 1}\) consists of all spoiler agents.
		From the definition of the preference lists \(\Pi\) is a Nash stable outcome for \(\JJJ\):

		By definition none of the normal or marker agents admits a Nash-deviation from \(\Pi\), as they are in their most preferred types of coalitions.

		It remains to consider the spoiler agents.
		By construction a spoiler agent can only admit a deviation to coalition \(C_i\) if
		\begin{enumerate}
		 \item[(a)] \(\frac{z_i + 1}{z_i + |\pi^{-1}(i)| + 1} = \frac{1}{b + 1}\) for some \(b \in \mathbb{N}\) or
		 \item[(b)] \(\frac{z_i + 1}{z_i + |\pi^{-1}(i)| + 1} = \frac{r + 1}{r + b + 1}\) with \(\frac{z_i}{z_i + |\pi^{-1}(i)|} = \frac{r}{r + b} = \frac{z_j}{z_j + t}\) for some \(j \in [|A|]\), \(r \in [75j + 1]\) and odd \(t \in \{2s|A|+1,\ldots, 2s(|A| + 1)-1\}\).
		\end{enumerate}

		Assume that (a) occurs, i.e.\ there is \(i \in [|A|]\) such that \(\frac{z_i + 1}{z_i + |\tilde{\pi}^{-1}(i)| + 1} = \frac{1}{b + 1}\) for some \(b \in \mathbb{N}\).
		This is equivalent to \((z_i + 1)b = |\tilde{\pi}^{-1}(i)|\) for some \(b \in \mathbb{N}\).
		In particular \(|\tilde{\pi}^{-1}(i)|\) has to be even which contradicts the fact that \(\pi\) is a solution for \(\III\) in which all approved group sizes are odd.

		Now assume that (b) occurs.
		\[\frac{z_i}{z_i + |\pi^{-1}(i)|} = \frac{z_j}{z_j + t}\text{ implies } z_i = \frac{|\pi^{-1}(i)|}{t}z_j.\]
		Because \(|\pi^{-1}(i)|,t \in \{2s|A|+1,\ldots, 2s(|A| + 1)-1\}\), we have
		\[z_i = \frac{|\pi^{-1}(i)|}{t}z_j < \frac{2s(|A|+1)}{2s|A|}z_j = \frac{|A|+1}{|A|}z_j=(1+\frac{1}{|A|})z_j.\]
		Without loss of generality we assume that \(i \geq j\) (otherwise in an analogous way \(z_j < (1 + \frac{1}{|A|})z_i\)).
		By definition of \((z_i)_{i \in [|A|]}\) this implies, that, if \(i \neq j\), then
		\[100 \leq z_i - z_j < \frac{z_j}{|A|} < 100;\]
		a contradiction.
		Hence \(i = j\) and \(|\pi^{-1}(i)| = t\).
		In particular this implies \(r \leq 75i + 1 < z_i\).

		Although it is a little tedious to compute, one can also verify that
		\[r =  \frac{\displaystyle 1 - \frac{r + 1}{r + b + 1}}{\displaystyle \frac{r + 1}{r + b + 1}\left(1 + \frac{\displaystyle 1 - \frac{r}{r + b}}{\displaystyle\frac{r}{r + b}}\right) - 1}.
		\]
		This in turn is equal, by replacing the corresponding fractions, to
		\[\frac{\displaystyle 1 - \frac{z_i + 1}{z_i + |\pi^{-1}(i)| + 1}}{\displaystyle\frac{z_i + 1}{z_i + |\pi^{-1}(i)| + 1}\left(1 + \frac{\displaystyle 1 - \frac{z_i}{z_i + |\pi^{-1}(i)|}}{\displaystyle \frac{z_i}{z_i + |\pi^{-1}(i)|}}\right) - 1}.
		\]
		which can be seen to be equal to \(z_i\).
		This is a contradiction to \(r < z_i\).
	\end{claimproof}

	For the converse direction we assume that \(\JJJ\) is a \YES-instance and fix a solution \(\Pi\) that witnesses this.
	In the following claims we always speak about coalitions in \(\Pi\).

	\begin{claim}\label{claim:w1:goodratios}
		Every normal agent \(j\) with corresponding participant \(p\) in \(\III\) is in a coalition with red ratio \(\frac{z_i}{z_i + t}\) for some \(i\) and \(t\) for which \(a\) is the \(i\)-th activity in \(\III\) and \((a,t) \in S_p\).
	\end{claim}
	\begin{claimproof}
		Assume for contradiction that some agent \(j\) corresponding to participant \(p\) in \(\III\) is in a coalition \(C\) with red ratio not in
		\[\Big\{\frac{z_i}{z_i + t} \,\Big\vert\, (a,t) \in S_p \land a \text{ is the } i\text{-th activity}\Big\}.\]
		By the definition of the preferences, this means that \(j\) admits a Nash-deviation unless \(C\) is a coalition with only blue vertices.

		Now consider the spoiler agents.
		If one of them is in a coalition with no blue agents at all they admit a Nash-deviation to \(C\).
		Hence all spoiler agents are in coalitions with blue agents with red ratios that all contained blue agents in these coalitions prefer over \(0\), and the contained spoiler agents prefer at least as much as joining coalition \(C\).

		This means that all red ratios of the coalitions containing spoiler agents are of the form \(\frac{z_i}{z_i + t} = \frac{1}{b + 1}\) where \(i \in [|A|]\), \(t \in \{2s|A|+1,\ldots,2s(|A| + 1)-1\}\) is odd and \(b \in \mathbb{N}\).

		This is equivalent to $z_ib = t$ and implies that
		\[b \in \left\{\frac{2s|A|+1}{z_i},\ldots,\frac{2s(|A| + 1)-1}{z_i}\right\}.\]
		Because there are only \(2s(|A| + 1)|A|\) blue agents there can be only
		\begin{align*}
			\frac{2s(|A| + 1)|A|z_i}{2s|A|} &= (|A| + 1)z_i\\
			&\leq (|A| + 1)(100|A|+ 1) \\
			&\leq 200|A|^2
		\end{align*} such coalitions (we use $|A| \ge 3$).

		We want to upper-bound the number of red agents in these coalitions.
		For this purpose denote by \(p\) the number of red agents and \(q\) the number of blue agents in one such coalition.
		\[\frac{p}{p + q} = \frac{1}{1 + b} \Leftrightarrow p = \frac{q}{b}.\]
		Using the lower bound on \(b\) from above we get that
		\[p \leq q\frac{z_i}{2s|A|}\]
		and using the number of all blue agents to bound \(q\)
		\[p \leq \frac{2s(|A| + 1)|A|z_i}{2s|A|} = (|A|+1)z_i \leq 200|A|^2.\]

		Together this means at most \((200|A|^2)^2\) red agents are in these kinds of coalitions but there are more than \((200|A|^2)^2\) spoiler agents, meaning not all spoiler agents can be in these kinds of coalitions.
	\end{claimproof}

	\begin{claim}\label{claim:w1:normal_agents_red}
			Whenever a set of normal agents \(J = \{j_1, \dotsc, j_\ell\}\), corresponding to participants \(p_1, \dotsc, p_\ell\), respectively, are together in the same coalition then the red ratio of the coalition is \(\frac{z_i}{z_i + t}\) for some \(i\) and \(t\) for which \(a\) is the \(i\)-th activity in \(\III\) and \((a,t) \in \bigcap_{k \in [\ell]} S_{p_k}\).
	\end{claim}
	\begin{claimproof}
		In a Nash-stable solution, \(j_1\) and \(j_2\) have to be in coalitions with red ratios \(\frac{z_i}{z_i + t_1}\) and \(\frac{z_j}{z_j + t_2}\) where \(a_1\) is the \(i\)-th activity in \(\III\) and \(a_2\) is the \(j\)-th activity in \(\III\) and \((a_1,t_1) \in S_{p_1}\) and \((a_2,t_2) \in S_{p_2}\).
		As they are in the same coalition, this means:
		\[\frac{z_i}{z_i + t_1} = \frac{z_j}{z_j + t_2}\text{ implying } z_i = \frac{t_1}{t_2}z_j.\]
		Because \(t_1,t_2 \in \{2s|A|+1,\ldots, 2s(|A| + 1)-1\}\), we have
		\[z_i = \frac{t_1}{t_2}z_j < \frac{2s(|A|+1)}{2s|A|}z_j = \frac{|A|+1}{|A|}z_j=(1+\frac{1}{|A|})z_j.\]
		Without loss of generality we assume that \(i \geq j\) (which also implies \(t_1 \geq t_2\)).
		By definition of \((z_i)_{i \in [|A|]}\) this implies, that, if \(i \neq j\), then
		\[100 \leq z_i - z_j < \frac{z_j}{|A|} < 100;\]
		a contradiction.
		Hence \(i = j\) and \(t_1 = t_2\).
		This argument can now inductively be extended for \(\ell > 2\) in the same way.
	\end{claimproof}

	\begin{claim}\label{claim:w1:normal_agents_coal}
		For every \(i \in [|A|]\) there is at most one coalition with normal agents whose red ratio is of the form \(\frac{z_i}{z_i + t}\) with \((a,t) \in S_p\) for every normal agent in that coalition and its corresponding participant \(p \in P\) where \(a\) is the \(i\)-th activity.
	\end{claim}
	\begin{claimproof}
		If there is a coalition \(C\) with red ratio \(\frac{z_i}{z_i + t}\) with \((a,t) \in S_p\) for every normal agent in that coalition where \(a\) is the \(i\)-th activity and with at most \(75i\) red agents then one can argue that there is at least one spoiler agent that admits a deviation to that coalition:
		Note that the red ratio of \(C\) and an additional spoiler agent is in \(\mathfrak{S}\) by assumption.

		For Claim 4 we already argued that there are at most \((200|A|^2)^2\) spoiler agents in coalitions with blue-up-to-1 red ratios.
		Now we consider coalitions with small-split red ratios (ratios in $\mathfrak{S}$).

		The same bound of \(200|A|^2\) on the number of these coalitions follows exactly in the same way as for blue-up-to-1 red ratio coalitions after noting that \(z_iq = pt \geq t\) where \(p\) is the number of red and \(q\) is the number of blue agents in the coalition:
		\[\frac{z_i}{z_i + t} = \frac{p}{p + q}\]
		implies that \(z_iq \geq t\) and hence \(q \geq \frac{2s|A| + 1}{z_i}\).
		By exactly the same calculation as in \Cref{claim:w1:goodratios} there can only be at most \(\frac{2s(|A| + 1)|A|z_i }{2s|A|} \leq 200|A|^2\) coalitions with this minimum number of blue agents as soon as \(|A| \geq 3\).

		Also in a similar way as blue-up-to-1 red ratio coalitions, coalitions with small-split red ratio contain at most \(200|A|^2\) red agents:
		Again denoting the number of red agents in a coalition with small-split red ratio by \(p\) and of blue agents by \(q\) we get
		\[p = \frac{z_iq}{t}\]
		and using \(2s|A|\) as lower bound for \(t\),
		\[p \leq q\frac{z_i}{2s|A|},\]
		and using the number of all blue agents to bound \(q\)
		\[p \leq z_i\frac{2s(|A| + 1)|A| + |A|}{2s|A|} \leq 200|A|^2\]

		There are \((400|A|^2)200|A|^2 + 1 > (400|A|^2)200|A|^2\) spoiler agents.
		Hence there is at least one spoiler agent that would want to deviate to \(C\).

		This shows that all coalitions with red ratio of the form \(\frac{z_i}{z_i + t}\) contain at least \(75i + 1\) red agents.
		This means that in these coalitions there are \(r \geq 75i + 1\) red and \(b\) blue agents, where
		\[\frac{z_i}{z_i + t} = \frac{r}{r + b} \text{ which is equivalent to } z_ib=tr.\]
		Then
		\[b=\frac{tr}{z_i} \geq \frac{75i + 1}{100i + 1}t \geq \frac{3}{2}s|A|.\]

		If there are at least two coalitions with these forms of red ratios for any single \(i \in [|A|]\), then there are a total of at least \(3s|A|\) blue agents in these coalitions, which is more than \(2s(|A| + 1)\) for \(|A| > 2\).
		However by construction only \(2s(|A| + 1)-1\) blue agents prefer these red ratios over red ratio \(0\).
		(To observe this also recall the arguments of \Cref{claim:w1:normal_agents_red} to show that for \(i \neq j\) and appropriate \(t_1\) and \(t_2\), \(\frac{z_i}{z_i + t_1} \neq \frac{z_j}{z_j + t_2}\).)
		This is a contradiction.
	\end{claimproof}

	Finally we can show that \(\III\) and \(\JJJ\) are indeed equivalent.
	
	\begin{claim}
		If \(\JJJ\) is a \YES-instance, then \(\III\) is also a \YES-instance.
	\end{claim}
	\begin{claimproof}
		For \(i \in [|A|]\), using \Cref{claim:w1:normal_agents_coal}, let \(C_i\) be the unique coalition in \(\Pi\) (if there is any) with red ratio of the form \(\frac{z_i}{z_i + t}\) with \((a,t) \in S_p\) for every normal agent in \(C_i\) and its corresponding participant \(p \in P\) where \(a\) is the \(i\)-th activity.

		We define \(\pi(p) = a\) such that the agent corresponding to \(p\) is in coalition \(C_i\) in \(\Pi\) and \(a\) is the \(i\)-th activity in \(A\).
		By \Cref{claim:w1:goodratios} this defines \(\pi\) on all elements of \(P\).
		Moreover, \Cref{claim:w1:normal_agents_red} immediately implies that \(\pi\) is a solution for \(\III\).
	\end{claimproof}

\end{proof}

\section{Restricting Palette}\label{sec:ownHDGs}

In their recent work, Boehmer and Elkind~\cite{BoehmerE2020} also introduced a restriction to hedonic diversity games where an agent $i$'s preferences only depend on the ratio between $i$'s color and the size of the coalition---in other words, which other colors occur in the coalition does not matter for $i$. Computing stable outcomes for these ``\emph{own-color}'' hedonic diversity games (hereinafter \textsc{Own-HDG}) is a special case of the general \HDGshort problem considered here, notably the case where the preferences of agents only depend on one element of the palette. As a consequence, every algorithmic result obtained for \HDGshort\ in this paper immediately carries over also to \textsc{Own-HDG}; however, the same is not true for algorithmic lower bounds. In fact, as our final result we give a concrete example where the complexity of \HDGshort\ differs from that of \textsc{Own-HDG}: While \HDGshort\ is \NP-hard when restricted to instances with $\coalnum=2$ (and even $\scoalnum=3$) as per \Cref{lem:allHDG:NPh:TypesNum}, parameterizing by the number of coalitions yields a non-trivial dynamic-programming based algorithm for \textsc{Own-HDG-Nash}.

\newcommand{\siz}{\texttt{size}}
\newcommand{\alo}{\texttt{alloc}}
\newcommand{\new}{\texttt{new}}
\begin{theorem}
	\textsc{Own-HDG-Nash} is in \XP\ parameterized by $\coalnum$.
\end{theorem}
\begin{proof}
	Consider an instance $\III=(N=[n], (D_1,\dots,D_\gamma),$ $\{\succeq_i~|~i\in N\},$ $\coalsize, \coalnum, \scoalnum)$, and let $k=\coalnum$ denote our parameter. We begin by branching over the total size of each (possibly) non-trivial coalition $C_1,\dots,C_k$, i.e., we branch over all of the at most $n^k$-many functions $\siz \colon [k]\rightarrow [n]$; let us now fix one such function $\siz$. Observe that an optimal solution may consist of less than $k$ non-trivial coalitions, and this is reflected in $\siz$ also including $1$ in its range; what is important though is that every coalition not described by $\siz$ can contain only a single agent.
	
	On a high level, the core of the algorithm is a dynamic programming procedure which processes colors one by one, and at each step uses branching to decide how agents of the given color will be allocated to the various coalitions. To visualize the state space of the dynamic program, it will be useful to imagine it as an auxiliary digraph. Let $\alo: [k]\to [n]_0$ be a function with the same domain and similar range as \siz, and let a \emph{record} be a tuple of the form $(i\in [\gamma]_0, \alo)$. Intuitively, the first component of the record stores how many colors we have processed so far, and the second component stores how many agents of previous colors we have already assigned to the (possibly) non-trivial coalitions $C_1,\dots,C_k$.
	
	Observe that the number of records is upper-bounded by $n^{k+1}$. Let us call the record $(0,\{i\in [k] \mapsto 0\})$ the \emph{initial} record, and the record $(\gamma,\{i\in [k] \mapsto \siz(i)\})$ the \emph{target record}. We also immediately discard all records such that $\alo(i)>\siz(i)$ for any $i\in [k]$. Let us now consider the directed graph $G$ whose vertices are all remaining records and where there is an arc from a record $(i,\alo_1)$ to another record $(i+1,\alo_2)$ if and only if it is possible to add agents of color $i+1$ to the partial assignment captured by $\alo_1$ in a way where (1) all agents of color $i+1$ do not have Nash deviations, and (2) the resulting assignment is captured by $\alo_2$. Note that here we make use of the own-color property of the instance: to determine whether an agent has a Nash deviation, all we need to know is the total size of the coalition and the number of agents of the same color---which other colors and agents are assigned to the same coalition later does not matter.
	
	Let us now formalize the construction of the arcs in the state space digraph $G$. For each pair of records $(i,\alo_1)$ and $(i+1,\alo_2)$ and each $j\in [k]$, let $\new(j)=\alo_2(j)-\alo_1(j)$. If $\new(j)$ is negative for any $j$, there will be no arc between this pair of records and we continue with another pair. Otherwise, construct the following instance of \textsc{Network Flow on Bipartite Graphs} (with sink capacities). On one side we have coalitions $C_1,\dots,C_k$ as well as $C_0$, which is a dummy element representing agents assigned to other trivial coalitions; all of these elements will be sinks, each $C_j$, $j\in [k]$ will have a capacity of $\new(j)$, and $C_0$ will have a capacity of $|D_i|-\sum_{j\in [k]}\new(j)$. On the other side we have one source element per agent in $D_i$, each with a supply of $1$, where an agent $a\in D_i$ has an edge to coalition $C_j$ if
	\begin{itemize}
	\item for $j=0$: $a$ must weakly prefer being alone to any of the coalitions $C_1,\dots,C_k$---in particular, for each $C_\ell$ the fraction $\frac{\new(\ell)+1}{\siz(\ell)+1}$ must not be preferred to the fraction $\frac{1}{1}$.
	\item for $j>0$: $a$ must weakly prefer being in coalition $C_j$ to being alone or in any other of the coalitions---in particular, $\frac{\new(j)}{\siz(j)}$ must be weakly preferred to both $\frac{1}{1}$ and $\frac{\new(\ell)+1}{\siz(\ell)+1}$ for all $\ell\in [k]\setminus \{j\}$.
	\end{itemize}
	
	We then find a maximum flow in this instance of \textsc{Network Flow on Bipartite Graphs} in time $n^{\bigoh(1)}$, and we add the arc from $(i,\alo_1)$ to $(i+1,\alo_2)$ if and only if there is a flow which routes all the agents to a sink.
	
	Once we have constructed all the arcs in $G$ using the above procedure, the algorithm simply runs Dijkstra's algorithm to find out whether there is a path from the initial record to the final record. If yes, we output ``\YES'', and otherwise we output ``\NO''. The running time is upper-bounded by $n^{\bigoh(k^2)}$: the graph $G$ contains $n^{\bigoh(k)}$-many vertices, for each pair of vertices we need to determine whether there is an edge (a subprocedure which requires time at most $n^{\bigoh(1)}$), and Dijkstra's algorithm can be implemented in time $\bigoh(|V(G)|^2)\in n^{\bigoh(k^2)}$.
	
	Finally, we argue correctness. Assume that $\III$ was a \YES-instance, as witnessed by some outcome $\Pi$. Then consider the sequence of vertices in $G$ which starts from the initial record and proceeds from record $(i,\alo_1)$ to record $(i+1,\alo_2)$ in a way where the number of new agents of color $i+1$ added to each coalition precisely matches the number of agents in these coalitions specified by $\Pi$. Since $\Pi$ is Nash-stable, the flow subprocedure described above will succeed and hence there will be an arc from $(i,\alo_1)$ to $(i+1,\alo_2)$. Hence, this sequence forms a path in $G$ from the initial record to the target record, and so our algorithm must also output \YES.
	
	On the other hand, assume our algorithm outputs \YES. Then there must exist a path in $G$ from the initial record to the target record, and each $i$-th arc on this path specifies a certain allocation of the agents of color $i$ to the coalitions. The construction of $G$ guarantees that such an allocation is Nash stable, and hence the path as a whole witnesses the existence of a Nash-stable outcome for $\III$.
\end{proof}

\section{Conclusions and future work}

In our work, we provided a complete complexity picture for \allHDGNashShort and \allHDGIndividualShort with multiple colors. Since the problem is \NPc in the general setting, we incorporated the so-called parameterized complexity analysis to provide a full dichotomy between tractable and intractable cases with respect to the most natural parameters. In particular, it showed that most of the tractable fragments require the number of colors $\colors$ to be among the parameters. All algorithmic results are complemented by the respective hardness lower bounds. Hence, no parameter can be excluded from the tractable fragments. On the other hand, it is likely that the running time of some of the presented algorithms can be improved by using more advanced techniques, as our intention was to provide problem classification into complexity classes. In addition, it can be the case that some of our algorithms cannot be improved under reasonable theoretical assumptions such as the well-known Exponential Time Hypothesis~(ETH). At the same time, our hardness reductions are not necessarily tight regarding constants used in the constructions. This can also be an interesting direction for future research.

We focused on two specific stability notations -- Nash and individual stability. It is natural to ask what is the complexity of \HDGshort with respect to other stability notations. Hedonic games can be also studied from the viewpoint of voting theory and one can ask what is the complexity of the determination of an outcome that maximizes social welfare. As mentioned earlier in our related work section, this direction of research was already investigated for hedonic diversity games by Darmann~\cite{Darmann21}. However, he focused on the case with $2$ colors and studied the problem with dichotomous preferences. 

In our work, we studied \textsc{HDGs} from the viewpoint of natural parameters. There are numerous other interesting parameters arising from both real-world and theoretical perspective. One such additional parameterization can be combined lower-bound and upper-bound on the coalition size as known from, e. g., \textsc{Group Activity Selection Problem}~\cite{DarmannDDLS17}.

Finally, in \Cref{sec:ownHDGs} we proposed the \textsc{Own-HDG} model and showed that the lower-bounds for this problem are not necessary the same as for the \textsc{All-HDG} model. It is natural to ask, where is the exact boundary between tractable and intractable cases for this simplified model.

\section*{Acknowledgements}
All authors are grateful for support from the OeAD bilateral Czech-Austrian WTZ-funding Programme (Projects No. CZ 05/2021 and 8J21AT021).

Robert Ganian and Thekla Hamm acknowledge support from the Austrian Science Foundation (FWF, project Y1329). Thekla Hamm also acknowledges support from FWF, projects W1255-N23 and J4651-N. 

Dušan Knop, Šimon Schierreich, and Ondřej Suchý acknowledge the support of the Czech Science Foundation Grant No. 22-19557S. Šimon Schierreich was additionally supported by the Grant Agency of the Czech Technical University
in Prague, grant No. SGS23/205/OHK3/3T/18.

\bibliography{references}

\end{document}